%% file: main.tex
\documentclass[10pt, conference]{IEEEtran}

\input{macros}

\makeatletter
\newcommand{\linebreakand}{%
  \end{@IEEEauthorhalign}
  \hfill\mbox{}\par
  \mbox{}\hfill\begin{@IEEEauthorhalign}
}
\makeatother

\begin{document}


\title{Specification Mining for Smart Contracts with Trace Slicing and Predicate Abstraction}

\IEEEoverridecommandlockouts 

\author{\IEEEauthorblockN{Ye Liu$^1$\thanks{This work was done while Ye Liu was a student at Nanyang Technological University.}}
\IEEEauthorblockA{\textit{Nanyang Technological University} \\
Singapore \\
li0003ye@e.ntu.edu.sg}
\and
\IEEEauthorblockN{Yixuan Liu}
\IEEEauthorblockA{\textit{Nanyang Technological University} \\
Singapore\\
liuy0255@e.ntu.edu.sg}
\linebreakand
\IEEEauthorblockN{Yi Li}
\IEEEauthorblockA{\textit{Nanyang Technological University} \\
Singapore\\
yi\_li@ntu.edu.sg}
\and
\IEEEauthorblockN{Cyrille Artho}
\IEEEauthorblockA{\textit{KTH Royal Institute of
Technology} \\
Sweden \\
artho@kth.se}}

\maketitle






\begin{abstract}
	Smart contracts are computer programs running on blockchains to implement Decentralized
	Applications.
	The absence of contract specifications hinders routine tasks, such as contract understanding and
	testing.
	In this work, we propose a specification mining approach to infer contract specifications from
	past transaction histories.
	Our approach derives high-level behavioral automata of function invocations,
	accompanied by program invariants statistically inferred from the transaction histories.
	We implemented our approach as tool \tool and evaluated it on eleven well-studied Azure benchmark smart contracts and six popular real-world DApp
	smart contracts.
	The experiments show that \tool mines reasonably accurate specifications that
	can be used to enhance symbolic analysis of smart contracts achieving higher code coverage and up to 56\,\% speedup, and facilitate DApp developers in maintaining high-quality documentation and test suites.
\end{abstract}

\input{sections/Introduction}
\input{sections/Background}
\input{sections/Motivation}
\input{sections/Framework}
\input{sections/Evaluation}

\input{sections/RelatedWork}
\input{sections/Conclusion}

\section*{Acknowledgement}
We thank all the anonymous reviewers for their constructive feedback on this work.
This work was supported by the Nanyang Technological
University Centre for Computational Technologies in Finance
(NTU-CCTF). Any opinions, findings, and conclusions or
recommendations expressed in this material are those of the
author(s) and do not necessarily reflect the views of NTU-CCTF.

\bibliographystyle{IEEEtran}
\balance
\bibliography{references}

\cleardoublepage
\input{sections/Illustration}

\end{document}

%% file: macros.tex
\usepackage{graphicx}

\usepackage{textcomp}
\usepackage{xcolor}
\def\BibTeX{{\rm B\kern-.05em{\sc i\kern-.025em b}\kern-.08em
    T\kern-.1667em\lower.7ex\hbox{E}\kern-.125emX}}
\usepackage[utf8]{inputenc}

\usepackage{xspace}
\usepackage{xcolor}
\usepackage{graphicx}
\usepackage[normalem]{ulem} 
\usepackage{enumitem}
\usepackage{microtype}
\usepackage{hyphenat}

\usepackage{amsthm}

\usepackage{booktabs}
\usepackage{multirow}
\usepackage{makecell}
\usepackage{ragged2e}
\usepackage{longtable}
\usepackage{diagbox}
\usepackage{caption}
\usepackage{subcaption}
\usepackage{wrapfig}
\usepackage{listings}

\usepackage[many]{tcolorbox}
\usepackage{lipsum}

\usepackage{bussproofs}
\usepackage{stackengine}

\usepackage{pifont}
\usepackage{tikz}
\usetikzlibrary{calc}
\usetikzlibrary{shapes.geometric}
\usetikzlibrary{decorations.pathreplacing}
\usetikzlibrary{positioning}
\usetikzlibrary{calc}
\usetikzlibrary{arrows}
\usetikzlibrary{shadows}

\usepackage{balance}
\usepackage{datetime} 

\usepackage{hyperref}
\usepackage[capitalize]{cleveref}

\usepackage{algorithm}
\usepackage[noend]{algpseudocode}

\crefname{section}{Sect.}{Sects.}
\Crefname{section}{Section}{Sections}

\newcommand{\XSpace}[1]{}
\newcommand{\XComment}[1]{}
\usepackage{todonotes} 

\newcommand{\revise}[1]{\textcolor{black}{#1}}

\newcommand{\DefMacro}[2]{\expandafter\newcommand\csname rmk-#1\endcsname{#2}}
\newcommand{\UseMacro}[1]{\csname rmk-#1\endcsname}

\renewcommand{\paragraph}[1]{\vskip 0.05in \noindent\textbf{#1.}}

\newcommand{\InputWithSpace}[1]{\bgroup\def\arraystretch{1.1}\input{#1}\egroup}
\newcommand{\Code}[1]{{\ifmmode{\mathtt{#1}}\else$\mathtt{#1}$\fi}}
\newcommand{\CodeIn}[1]{{\ifmmode{\mathtt{#1}}\else$\mathtt{#1}$\fi}}

\newcommand{\tool}{\textsc{SmCon}\xspace}

\definecolor{silver}{RGB}{192,192,192}

\newcolumntype{R}[1]{>{\RaggedLeft\arraybackslash}p{#1}}
\newcolumntype{L}[1]{>{\RaggedRight\arraybackslash}p{#1}}

\DefMacro{NumOfSubjectsReduceOver90Percent}{seven}

\newtcolorbox{rqbox}[1]{%
    tikznode boxed title,
    enhanced,
    arc=0mm,
    boxrule=0.5pt,
    interior style={white},
    attach boxed title to top center= {yshift=-\tcboxedtitleheight/2},
    colbacktitle=white,coltitle=black,
    boxed title style={size=normal,colframe=white,boxrule=0pt},
    title=\textbf{Answer to }{\textbf{#1}}}

\newcommand{\join}{\oplus}

\definecolor{codegreen}{rgb}{0,0.6,0}
\definecolor{codegray}{rgb}{0.5,0.5,0.5}
\definecolor{codepurple}{rgb}{0.58,0,0.82}
\definecolor{backcolour}{rgb}{0.95,0.95,0.92}

\usepackage{alphabeta}

\usepackage[group-separator={,}]{siunitx}

\usepackage{graphicx}
\usepackage{booktabs}
\usepackage{multirow}
\usepackage[flushleft]{threeparttable}
\usepackage{xcolor}
\usepackage{syntax}

\usepackage{changes}
\usepackage{blkarray, bigstrut}
\newcommand{\setBAcolsep}[1]{\BA@colsep=#1}
\usepackage{algorithm}
\usepackage[noend]{algpseudocode}
\usepackage{algorithmicx}

\usepackage{pgfplotstable}

\usepackage{amsmath,amssymb,amsfonts}

\usepackage{soul}

\usepackage{bussproofs}

\usepackage{colortbl}

\input{solidity-highlighting}

\newcommand{\website}{{https://sites.google.com/view/smcon/}}

\usepackage{tikz}
\usetikzlibrary{automata, positioning, arrows, shapes}

\tikzstyle{block} = [rectangle, draw, text width=5em,
text centered, rounded corners,
minimum height=4em, node distance=3cm]
\tikzstyle{line} = [draw, -latex']
\tikzstyle{cloud} = [draw, ellipse, node distance=2.5cm, minimum height=2em]
\tikzstyle{blank} = [node distance=1cm]

\newcommand{\lts}{\ensuremath{\mathit{LTS}}\xspace}

\newcommand{\efsm}{\ensuremath{\mathit{EFSM}}\xspace}

\newcommand{\dicether}{Dicether\xspace}

\renewcommand{\#}{\texttt{\symbol{`\#}}}
\usepackage{bussproofs}

\newcommand{\history}{\ensuremath{\mathit{LTS_h}}\xspace}

\newcommand{\efsmtext}{extended finite state machine\xspace}
\newcommand{\cegar}{counterexample-guided abstraction refinement\xspace}
\newcommand{\Cegar}{Counterexample-guided abstraction refinement\xspace}
\newcommand{\minimalabstraction}{minimal existential abstraction\xspace}

\newcommand{\erc}{ERC-20\xspace}
\newcommand{\gamechannel}{\texttt{GameChannel}\xspace}

\usepackage{listings,multicol}  
\usepackage{filecontents}       

\newsavebox{\mybox}
\newtheorem{definition}{Definition}
\newtheorem{theorem}{Theorem}

\newcommand{\onetail}{\textsc{1-Tail}\xspace}
\newcommand{\twotail}{\textsc{2-Tail}\xspace}
\newcommand{\onesekt}{\textsc{SEKT-1}\xspace}
\newcommand{\twosekt}{\textsc{SEKT-2}\xspace}
\newcommand{\contractor}{\textsc{Contractor++}\xspace}
\newcommand{\code}[1]{\emph{#1}}

%% file: solidity-highlighting.tex
\usepackage{listings, xcolor}

\definecolor{verylightgray}{rgb}{.97,.97,.97}

\lstdefinelanguage{Solidity}{
	keywords=[1]{anonymous, assembly, assert, balance, break, call, callcode, case, catch, class, constant, continue, constructor, contract, debugger, default, delegatecall, delete, do, else, emit, event, experimental, export, external, false, finally, for, function, gas, if, implements, import, in, indexed, instanceof, interface, internal, is, length, library, log0, log1, log2, log3, log4, memory, modifier, new, payable, pragma, private, protected, public, pure, push, require, return, returns, revert, selfdestruct, send, solidity, storage, struct, suicide, super, switch, then, this, throw, true, try, typeof, using, value, view, while, with, addmod, ecrecover, keccak256, mulmod, ripemd160, sha256, sha3}, 
	keywordstyle=[1]\color{blue}\bfseries,
	keywords=[2]{address, bool, byte, bytes, bytes1, bytes2, bytes3, bytes4, bytes5, bytes6, bytes7, bytes8, bytes9, bytes10, bytes11, bytes12, bytes13, bytes14, bytes15, bytes16, bytes17, bytes18, bytes19, bytes20, bytes21, bytes22, bytes23, bytes24, bytes25, bytes26, bytes27, bytes28, bytes29, bytes30, bytes31, bytes32, enum, int, int8, int16, int24, int32, int40, int48, int56, int64, int72, int80, int88, int96, int104, int112, int120, int128, int136, int144, int152, int160, int168, int176, int184, int192, int200, int208, int216, int224, int232, int240, int248, int256, mapping, string, uint, uint8, uint16, uint24, uint32, uint40, uint48, uint56, uint64, uint72, uint80, uint88, uint96, uint104, uint112, uint120, uint128, uint136, uint144, uint152, uint160, uint168, uint176, uint184, uint192, uint200, uint208, uint216, uint224, uint232, uint240, uint248, uint256, var, void, ether, finney, szabo, wei, days, hours, minutes, seconds, weeks, years},	
	keywordstyle=[2]\color{teal}\bfseries,
	keywords=[3]{block, blockhash, coinbase, difficulty, gaslimit, number, timestamp, msg, data, gas, sender, sig, value, now, tx, gasprice, origin},	
	keywordstyle=[3]\color{violet}\bfseries,
	identifierstyle=\color{black},
	sensitive=false,
	comment=[l]{//},
	morecomment=[s]{/*}{*/},
	commentstyle=\color{gray}\ttfamily,
	stringstyle=\color{red}\ttfamily,
	morestring=[b]',
	morestring=[b]"
}

%% file: sections/Introduction.tex
\section{Introduction}

Blockchain technology has developed rapidly in recent years, since the introduction of
Bitcoin~\cite{nakamoto2008bitcoin} by Nakamoto in 2008.
Blockchain itself is a distributed ledger maintained and shared by a peer-to-peer (P2P) network,
and it evolved into a platform which supports the deployment and execution of smart contracts,
popularized by Ethereum~\cite{Ethereum}.
Smart contracts are self-executing computer programs used to implement Decentralized Applications
(DApps).
Users interact with smart contracts by executing transactions on the blockchain.
Ethereum, the most prominent smart contract platform, is empowering many DApps, spanning areas
such as
finance, health, governance, games, etc.~\cite{dappradar}.
As of May 2023, there are more than 50 million smart contracts deployed on Ethereum, and these smart contracts have supported 13,968 DApps~\cite{Etherscan,dappradar}.

Despite the high stakes involved, smart contracts are often developed in an undisciplined way.
The existence of bugs and vulnerabilities compromises the reliability and security of smart
contracts and endangers the trust of users.
Durieux et~al.~\cite{durieux2020empirical} reported that nearly 10\% of the smart contracts may
contain security vulnerabilities related to access controls.
\erc~\cite{eip20} is the most popular smart contract standard on Ethereum, yet 13\% of the \erc
token contracts do not conform to the standard specification~\cite{chen2019tokenscope}.
Moreover, Qin et~al.~\cite{qin2021attacking} demonstrated how economic behavior models can be
exploited to attack the DeFi ecosystem with flash loans.
A major difficulty in validating the conformance of smart contracts, i.e., whether the contract
implementation adheres to the expected behaviors, is the lack of documented formal specifications.

Formal specifications capture the expected contract behaviors, in terms of formal languages, based
on a formal model~\cite{jiao2020semantic} with precise semantics.
Specifications of a smart contract play a central role in describing, understanding, reasoning
about contract behaviors, and detecting, through testing and verification, non-conformance issues
such as functional bugs and security vulnerabilities.


Similar to traditional formal specifications, two forms of smart contract specifications have
been studied in past work:
(1) function-level program invariants~\cite{Liu2022IAD}, which are used in
testing~\cite{Wang2019VUL}, verification~\cite{permenev2020verx,liu2020towards}, and runtime
validation~\cite{li2020securing} of smart contracts; and (2) contract-level behavioral
specifications in the form of automata~\cite{mavridou2018designing}, which can be used to support
contract synthesis~\cite{mavridou2018tool}, model-based testing~\cite{Liu2020MAM}, design verification~\cite{mavridou2019verisolid}, and workflow verification~\cite{wang2019formal}.
Specifically, Wang et al.~\cite{wang2019formal} performed workflow verification via semantic
conformance checking between state machine-based workflow specifications and smart contracts from
the \textit{Azure Blockchain Workbench}, an enterprise blockchain from Microsoft.

In this paper, we focus on mining high-level automata-based specifications automatically for smart
contracts.
Many approaches have been proposed for this task on traditional program traces: for example,
grammar inference techniques~\cite{aarts2012automata,clarke2000counterexample,de2010grammatical}
and deep learning-based techniques~\cite{le2018deep} have been used to learn automata from a set
of program execution traces.
The k-tail algorithm and its
variants~\cite{biermann1972synthesis,krka2014automatic,lorenzoli2008automatic} merge states if the
same set of ``tail'' invocation sequences are observed.

However, the way smart contracts behave poses new challenges for mining automata-based behavioral
models.
As they are usually deployed on public blockchain networks, smart contracts handle multiple user
interactions simultaneously.
Therefore, the execution traces recorded in contract transaction histories consist of interleaving
events triggered by different user interactions and may belong to different sessions.
Since there does not exist a standard approach for managing user sessions, the execution traces
cannot be easily separated for independent interactions.
Moreover, predicate abstraction is crucial in deriving compact but accurate automata.
Yet, the choice of predicates remains challenging and is often tightly tied with the specific
analysis tasks.
The predicate abstraction techniques used in computing state abstractions must be tailored to take
into account the specific data structures and runtime environments of smart contracts.
To mine more accurate automata specification efficiently for smart contracts, we propose a
specification mining algorithm powered by trace slicing and predicate
abstraction~\cite{graf1997construction}.
The contract specification mining process is preceded by a slicing of the transaction histories.
We perform trace slicing on the transaction histories via a parametric binding learned from the
existing test suites.
A slice of history is a sequence of inter-related transactions, e.g., all transactions related to
one specific trade session.
Smart contract transaction histories, being stored persistently on blockchain, record all past
function executions since the contract deployment.
To find suitable predicate candidates for state abstraction, we use a statistical inference
technique~\cite{daikon,Liu2022IAD} to generate a set of dynamic invariants, based on the
transaction histories.
Then, we follow the \cegar (CEGAR) approach~\cite{clarke2000counterexample} to perform a lazy state abstraction, and introduce
\minimalabstraction to ensure the automata specification is accurate and simple.
Finally, our automata specification subsumes all observed invocation sequences and at the same
time preserves its generality.


In summary, we make the following contributions.
First, we formalize the specification mining problem for smart contracts.
Second, we propose a CEGAR-based specification mining algorithm, powered by trace slicing and
predicate abstraction.
Third, we implement our approach in tool~\tool and evaluate it on eleven well-studied Azure benchmark smart contracts and six popular real-world DApp smart contracts.
The experiments indicate that the mined specifications are precise and useful for DApp development and can enhance symbolic analysis of smart contracts in achieving higher code coverage and detecting more issues within smaller number of function call sequences and speeding up symbolic execution by up to 56\,\% by enforcing trace slicing.
The benchmarks, raw results, and source code are available at: {\website}.

\paragraph{Organization}
The rest of the paper is organized as follows.
\Cref{sec:background} provides the background.
\Cref{sec:glance} illustrates our approach through an example.
\Cref{sec:algorithm} introduces our specification mining algorithm, followed by the implementation
and evaluation in \cref{sec:eval}.
We compare with the related work in \cref{sec:related} and conclude the paper in
\cref{sec:conclusion}.

%% file: sections/Background.tex
\section{Background}\label{sec:background}

We borrow terminology about (non-)parametric events and traces
from~\cite{lee2011mining}.
\begin{definition} [Non-Parametric Events and Traces]
  Let $\xi$ be a set of (non-parametric) events, called base events or simply events.
  An $\xi$-trace, or simply a (non-parametric) trace is any finite sequence of events in $\xi$,
  that is, an element in $\xi^\star$. If event $e \in \xi$ appears in trace $w \in \xi^\star$ then
  we write $e \in w$.
  \label{def:nonparametric}
\end{definition}

\begin{definition}[Parametric Events and Traces]
  Let $X$ be a set of parameters and let $V$ be a set of corresponding parameter values.
  If $\xi$ is a set of base events as in \cref{def:nonparametric}, then let $\xi(X)$ be the set of
  corresponding parametric events $e(\theta)$, where $e$ is a base event in $\xi$ and $\theta$ is a
  partial function in $[X \rightharpoondown V]$.
  A parametric trace is a trace with events in $\xi(X)$, that is, a word in $\xi(X)^\star$.
\end{definition}

%


From a user's perspective, a smart contract is a set of interface functions which can be invoked to
execute contract code.
Let these interface functions be represented as base events: $\xi$ is the set of interface function
names and $e \in \xi$ corresponds to a contract function.
The execution of $e$ accepts parameters (denoted as $X$), including the user-provided function
inputs ($X_1$) and the contract state variables ($X_2$) stored on the blockchain.
Let $V$ be the corresponding values of $X$ in parametric traces.
Let $D_X$, $D_{X_1}$, and $D_{X_2}$ be the corresponding domains.
Finally, given a smart contract, let $\xi(X)$ be the set of all function executions, and any
function invocation sequence can be represented as a parametric trace (word) in $\xi(X)^\star$.
The behaviors of a smart contract can be captured by a labeled transition system that accepts all
its function invocation sequences.


\begin{definition}[Labeled Transition System (\lts)~\cite{beillahi2020behavioral}]
	A smart contract is a labeled transition system $(\mathit{S}, \mathit{s}_0, \Sigma, \delta)$
	where $\mathit{S}$ is a set of possibly-infinite states, $S \subseteq D_{X_2}$ $\mathit{s}_0 \in
	S$ is an initial state,
	$\Sigma$ is a possibly-infinite alphabet, $\Sigma \subseteq \xi(X)^\star$, and $\delta \subseteq \mathit{S} \times \Sigma \times
	\mathit{S}$ is a set of transitions.
\end{definition}
An \lts can be represented more compactly by abstracting it into an \efsm.
\begin{definition}[Extended Finite State Machine (\efsm)~\cite{cheng1993automatic}]
\efsm is defined as a 6-tuple $(Q, q_0, \Sigma', G, U, T)$ where,
\begin{itemize}
	\item $Q$ is a finite set of symbolic states under a predicate abstraction $\alpha: S \rightarrow
	Q$,
  \item $q_0 \in Q$ is the initial symbolic state,
	\item $\Sigma'$ is a finite alphabet defined, $\Sigma' \subseteq \xi^\star$,
	\item $\mathit{G}$ is a set of \textit{guarding} function $\mathit{g_i}$ such that $g_i$: $D_X
	\rightarrow \{\mathit{True}, \mathit{False}\}$,
	\item $\mathit{U}$ is a set of \textit{update} function $u_i$ such that $u_i$: $D_X \rightarrow
	D_X$,
	\item ${T}$ is a transition relation, ${T}:  {Q} \times \mathit{G} \times \Sigma \rightarrow {U}
	\times Q $.
\end{itemize}
\end{definition}

To compute state abstractions, predicate abstraction~\cite{clarke2000counterexample} is typically
used, which is a function to create a partition of the domains of data types.
For example, the widely used predicate abstraction for integer domain is \{$\mathit{neg}$,
$\mathit{zero}$, $\mathit{pos}$\} which represent negative, zero and positive numbers respectively.
However, there could be many \efsm candidates that an \lts can be abstracted into.
In this paper, we borrow the concept of \emph{minimal existential
abstraction}~\cite{chauhan2002automated} and later use it to obtain a compact \efsm.
\begin{definition}[Minimal Existential Abstraction~\cite{chauhan2002automated}]
\label{def:minimal}
\efsm$=(Q,\allowbreak q_0, \allowbreak \Sigma', \allowbreak G, \allowbreak U, \allowbreak T)$ is
the \minimalabstraction of \lts$=(\mathit{S},
\mathit{s}_0, \Sigma, \delta)$ with respect to $\alpha: S \rightarrow Q$ iff,

\begin{small}
\vspace{-.05in}
\begin{equation}
\exists s_0\in S \cdot \alpha(s_0)=q \iff q = q_0
\end{equation}%
\vspace{-.25in}
\begin{multline}
\exists (s_0, e_0(\theta_0), s_1), \ldots, (s_{n-1}, e_{n-1}(\theta_{n-1}), s_n) \in \delta
\;\cdot \\
\alpha(s_0) = q_0 \land \alpha(s_1) = q_1 \land \cdots \land \alpha(s_{n-1}) = q_{n-1} \land
\alpha(s_{n}) = q_{n} \\ \iff (q_0, g_i,  e_0, u_i, q_1), \ldots, (q_{n-1}, g_j, e_{n-1}, u_j, q_n)
\in T
\end{multline}
\end{small}
\end{definition}

Intuitively, the \minimalabstraction implies that: (1) the initial concrete state can be mapped to
the initial symbolic state in the \efsmtext, and vice versa;
(2) every concrete path is preserved in the \efsmtext, and every symbolic path in the \efsmtext has
at least a corresponding concrete path.

%% file: sections/Motivation.tex
\section{Approach at a Glance}\label{sec:glance}

\begin{figure*}[t]
  \centering
  \includegraphics[width=\textwidth]{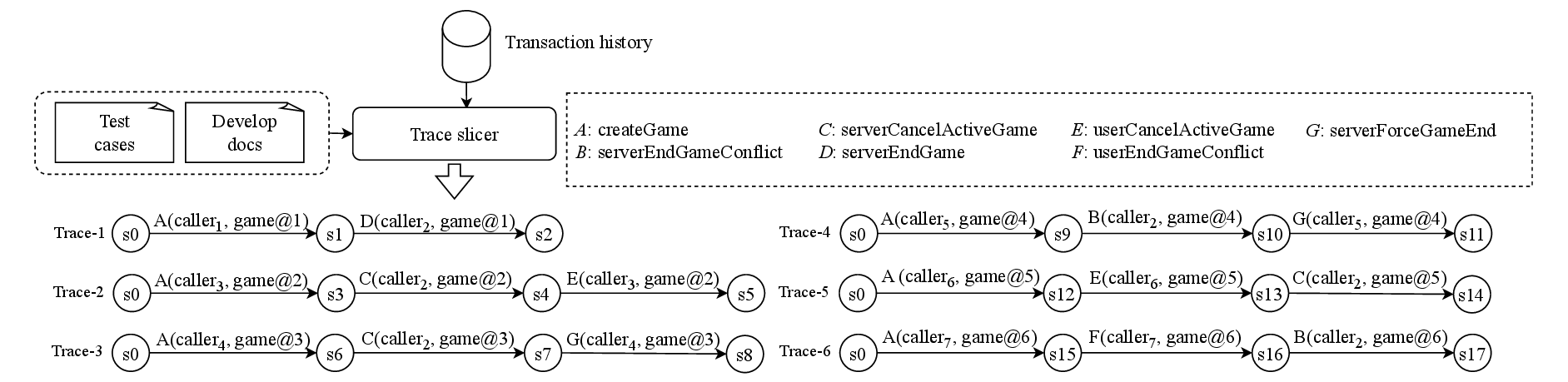}
  \caption{Six game invocation sequences for GameChannel.}\label{tab:Dicether-LOG}
\end{figure*}

We illustrate our approach using the \texttt{GameChannel} contract from a DApp called \dicether.
\dicether is a decentralized casino application on Ethereum, relying on a smart contract to provide
an open, secure, and fair gaming experience.
A new game is created by calling the contract function \texttt{createGame}.
When approaching the end of a game, an admin user may invoke the \texttt{serverEndGame} function to
close the game.
More details can be found in its \revise{development} documentation~\cite{dicether}.

\Cref{tab:Dicether-LOG} overviews how we separate interleaving interactions from past transaction
histories.
A transaction history is a sequence of transactions, where each transaction can be decoded as a
contract function invocation.
We apply a slicing function, which is determined by interaction patterns observed in test suites,
on the transaction history, to produce a set of independent invocation sequences.
For \gamechannel, there are six game interaction sequences, corresponding to six user sessions.
For instance, $user_1: \texttt{A(gameId:1)}$ indicates the invocation of \texttt{createGame} by
${user_1}$ for creating a game with index 1.
For simplicity, we omit the values of the other function parameters and the transaction environment
variables.
A function invocation may change the values of state variables, thus updating contract states.
In \cref{tab:Dicether-LOG}, the first game is created by ${user_1}$ and after a
while ended by ${user_2}$.
Three game states, $s_0$, $s_1$, and $s_2$, are involved.
The second game is created by ${user_3}$, and later canceled by ${user_2}$ and
${user_3}$ via \texttt{serverCancelActiveGame} and \texttt{userCancelActiveGame}, respectively.

From these invocation sequences, we can construct an extended finite state machine
annotated with function pre-/post-conditions, as a specification of the observed contract behaviors.
Specifically, in \gamechannel, each function pre-/post-condition consists of a set of predicates
either relevant to game state variables or function input parameters.
\Cref{fig:gamechannel-coredata} shows the data structure used in \gamechannel, where
\texttt{server}, \texttt{gameIdCntr}, and \texttt{gameIdGame} maintain information about the game
manager, the number of created games, and all game state information, respectively.
The game state variables include \texttt{status}, \texttt{roundId}, \texttt{endInitiatedTime}, and \texttt{stake}.
The variable \texttt{status} being ENDED (0) indicates that a game either has not been created or
has already been terminated;
\texttt{roundId} is an unsigned integer used to record the current game round;
\texttt{endInitiatedTime} records when a game is required to terminate itself as per users'
requests;
and \texttt{stake} keeps the amount of fund that a player deposits into the contract when creating
a game.
Since all parameter values, including contract state variables and user-provided function inputs,
can be decoded from blockchain transactions, we can infer dynamic invariants to be candidates of
predicates on function pre-/post-conditions.
\Cref{fig:partition} shows the 11 resulting predicates.

\begin{figure}[t]
{\small
\begin{align*}
& P_1: \mathit{status} = 0 \quad P_2: \mathit{status} = 1 \\
& P_3: \mathit{status} = 2 \quad P_4: \mathit{status} = 3 \quad P_5: \mathit{status} > 3\\
& P_6: \mathit{roundId} = 0 \quad P_7: \mathit{roundId} > 0\\
& P_8: \mathit{endInitiatedTime} = 0 \quad P_9: \mathit{endInitiatedTime} > 0 \\
& P_{10}: \mathit{stake} = 0 \quad P_{11}: \mathit{stake} >0
\end{align*}}%
\caption{The 11 predicates that partition the game state.}\label{fig:partition}
\end{figure}

Assume that all contract state variables are initialized to zero, so $s_0$ can be represented by
$P_1 \land P_6 \land P_8 \land P_{10}$.
These predicates also form the pre- and post-conditions in~\cref{tab:functionconditions}, where
some other parameter predicates in the pre- and post-conditions are over function input parameters,
i.\,e., ``\texttt{\_roundId}'' and caller of the function, i.\,e., ``\texttt{caller}''.
The precondition of the \textit{createGame} function is that all variable values, namely,
\texttt{status}, \texttt{stake}, \texttt{roundId}, and \texttt{endInitiatedTime}, are zero;
and its postcondition is that when \textit{createGame} finishes, the variable \texttt{status} is
set to \texttt{ACTIVE} (1) and the deposited \texttt{stake} is greater than zero and equals to the
transferred fund, i.\,e., $\mathit{msg.value}$.\footnote{In Solidity smart contracts,
$\mathit{msg.value}$ refers to the amount of transferred native cryptocurrency, e.g., ETH on
Ethereum, during contract function execution.}

\Cref{fig:dicether-spec} shows our mined automaton, of which we have confirmed the correctness using the ground truth specification of \gamechannel.
The mined automaton has seven symbolic states.
Only \texttt{createGame} can be called at the initial state ($q_0$).
Furthermore, when the caller is \texttt{server}, he/she is allowed to call \texttt{serverEndGame} to
terminate the game and move towards the final state ($q_6$) where \texttt{status} changes to be \texttt{ENDED}  (0).
Such an automaton captures the common usages of \gamechannel and its permission policies, thus being a \emph{likely} contract specification.

\begin{table*}[t]
\centering
\caption{The function pre-/post-conditions of \gamechannel.}\label{tab:functionconditions}
\setlength\tabcolsep{5pt}
\resizebox{.99\textwidth}{!}{
\begin{tabular}{lp{6.2cm}p{6.2cm}}
  \toprule
  Functions              & Preconditions                                              &
  Post-conditions                                              \\ \midrule
  createGame             & $P_1 \land P_6\land P_8 \land P_{10}$                            & $P_2
  \land P_6 \land P_8 \land P_{11} \land (\mathit{stake}=\mathit{msg.value})$ \\
  serverEndGameConflict  & $(P_2 \lor P_3) \land (\_\mathit{roundId} > 0) \land (\mathit{caller}=\mathit{server})$ & $P_1 \lor
  (P_4 \land P_7 \land P_9)  \land (\mathit{roundId} = \_\mathit{roundId} ) $                         \\
  serverCancelActiveGame & $(P_2 \lor (P_3 \land P_6)) \land (\mathit{caller}=\mathit{server})$            & $P_1 \lor
  (P_4 \land P_9)$                                  \\
  serverEndGame          & $P_2 \land (\mathit{caller}=\mathit{server})$                                 &
  $P_1$                                                        \\
  userCancelActiveGame   & $P_2 \lor (P_4 \land P_6)$                                    & $P_1 \lor
  (P_3 \land P_9)$                                  \\
  userEndGameConflict    & $(P_2 \lor P_4) \land (\_\mathit{roundId} > 0) $                      & $P_1 \lor
  (P_3 \land P_7 \land P_9 ) \land (\mathit{roundId} = \_\mathit{roundId} )$     \\
  serverForceGameEnd     & $P_4 \land (\mathit{caller}=\mathit{server})$                                 &
  $P_1$                                                        \\ \bottomrule
\end{tabular}}
\vspace{.1in}
\end{table*}

\begin{figure*}[t]
	\begin{subfigure}[b]{.34\textwidth}
		\centering
    \includegraphics[width=\columnwidth]{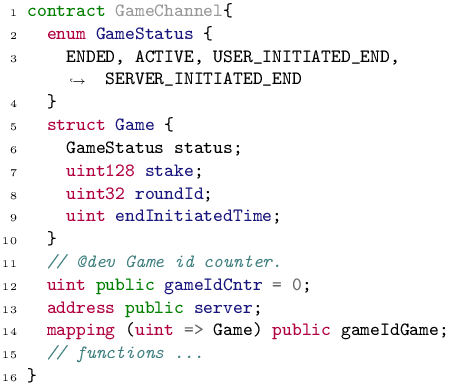}
		\caption{State variables of {GameChannel}.}
		\label{fig:gamechannel-coredata}
	\end{subfigure}
\hspace{-6mm}
	\begin{subfigure}[b]{.59\textwidth}
		\centering
    \includegraphics[width=\columnwidth]{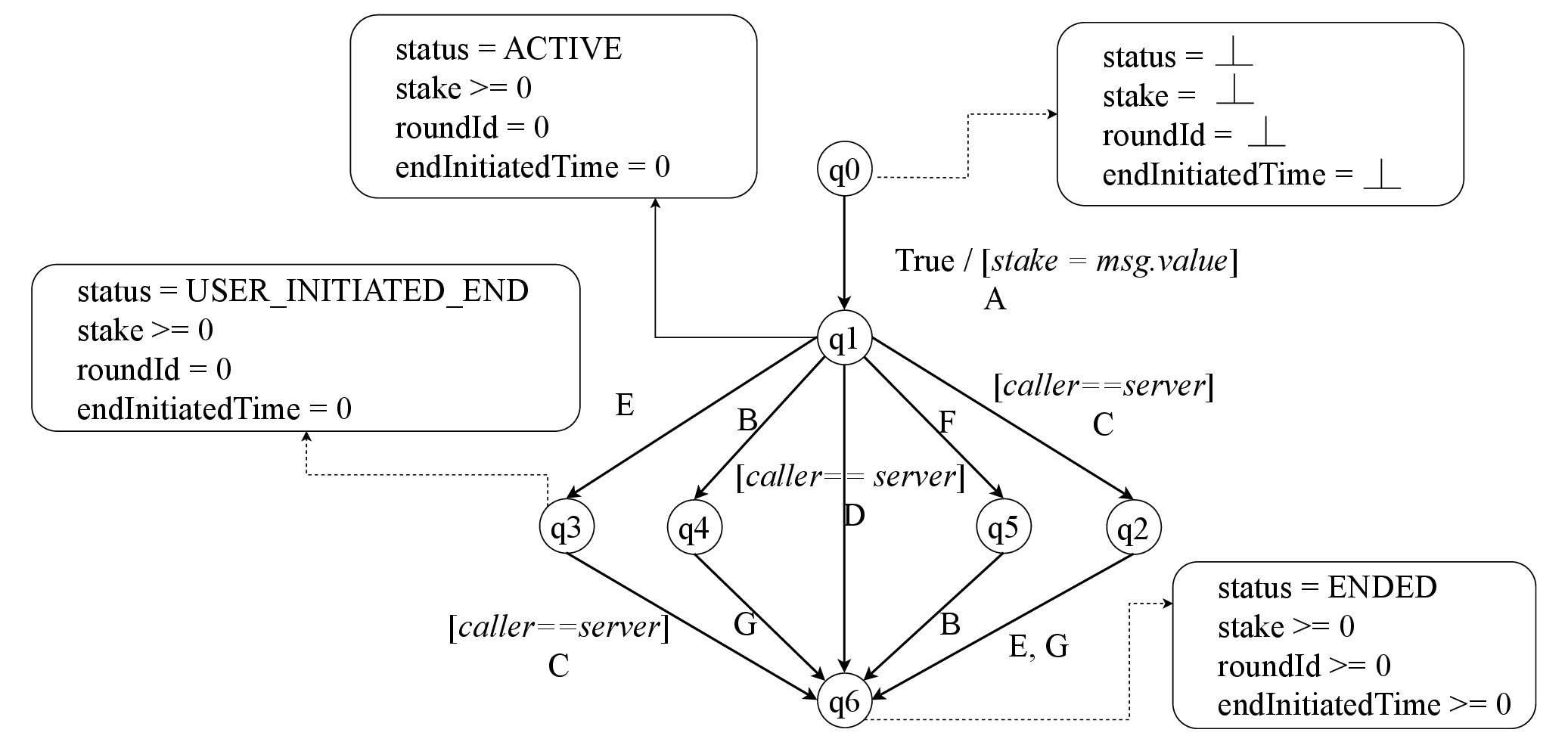}
		\caption{The mined automaton.}
		\label{fig:dicether-spec}
	\end{subfigure}
	\caption{The core data structure and the mined automaton of GameChannel.}
\end{figure*}


As for automata construction, \tool uses a CEGAR-like approach, which will be detailed
in \cref{sec:algorithm}.
Briefly, we perform a lazy abstraction, i.\,e., we do not refine predicate abstraction unless we have to.
To obtain an extended finite state machine,
\tool takes the sliced independent invocation sequences and the inferred function pre-/post-conditions as input.
Initially, we construct an automaton containing only two states and then revisit the automaton to recognize the spurious symbolic paths that have no support, i.\,e., a corresponding concrete invocation sequence in the past observations.
Then we refine the automaton to eliminate the spurious paths via either splitting larger states or removing unreachable transitions.
We repeat this process until no spurious path is included in the resulting automaton.

%% file: sections/Framework.tex
\section{Contract Specification Mining}\label{sec:algorithm}

In this section, we introduce the specification mining problem for smart contracts and present our
proposed algorithm.

\paragraph{Smart Contract Specification Mining}
Given a contract's transaction histories, where all the past contract behaviors are captured by
\history, the \emph{specification mining problem} is to mine an \efsm as the likely specification
of the smart contract.
To solve the specification mining problem, we first perform a \emph{trace slicing} on the input
transaction histories, to obtain multiple independent invocation traces.
Next, we find predicates that belong to preconditions or post-conditions of the smart contract's
functions.
Finally, we implement a \cegar loop to produce an \efsmtext, satisfying the minimal existential
abstraction property (see~\cref{def:minimal}).

\subsection{Trace Slicing}

Smart contracts are public-facing, and, by their nature, simultaneously accept inputs from multiple
users.
Contract executions in such a setting result in a linear transaction history, which consists of
interleaving execution traces triggered through multiple user interactions/sessions.
To record data owned by different users, most smart contracts supporting DApps, maintain a
collection of custom data objects, indexed by user(session)-specific parameters.
For example, the \texttt{GameChannel} contract maintains many concurrent game instances as state
variables.
To interact with a particular game instance, a user needs to specify the value of its
\texttt{gameId}, through input parameters of the transaction (see~\cref{tab:Dicether-LOG}).
To mine meaningful contract specifications from transaction histories with mixed interactions, one
has to slice them into independent traces for each game instance.

\begin{definition}[Trace Slicing~\cite{lee2011mining}]
	Given a parametric trace $\tau \in \xi(X)^\star$ and a parametric binding $\theta$ in
	$[X \rightharpoondown V]$, let the $\theta$-trace slice $\tau\upharpoonright_\theta  \in \xi
	^\star$ be the non-parametric trace defined as:
	\begin{itemize}
		\item $\epsilon \upharpoonright_\theta = \epsilon$, where $\epsilon$ is the empty trace/word, and
		\item $ (\tau e(\theta')) \upharpoonright_\theta =
		 \begin{cases}
			 \mbox{$(\tau \upharpoonright_\theta) e$, } & \mbox{if } \theta' \sqsubseteq \theta \\ \mbox{$\tau \upharpoonright_\theta$,} & \mbox{otherwise}
		\end{cases}
	$
	\end{itemize}
	where we say that $\theta'$ is less informative than $\theta$, written $\theta' \sqsubseteq
	\theta$ iff for any $x \in X$, if $\theta'(x)$ is defined then $\theta(x)$ is also defined and
	$\theta'(x) = \theta(x)$.
\end{definition}


%
%

A transaction history of smart contract can be seen as a parametric trace, and \emph{trace slicing}
slices the history into a set of independent invocation sequences via certain parametric
bindings (e.g., $\theta$)~\cite{lee2011mining}.
A trace slice $\tau \upharpoonright_\theta$ first filters out all the parametric events that are
irrelevant to the parameter instance $\theta$.
A trace slice also forgets the parameter bindings of parametric events.
As a result, a trace slice is non-parametric and merely a list of base events.
To find parametric bindings, we should first ascertain the relation between different events, or
say function invocations in smart contracts.
Such parametric bindings can be inferred from the existing DApp test suites, which demonstrate
typical usage scenarios and user interaction patterns.
Specifically, we may observe a group of related functions and what parameter values they share in a
unit test.
For example, the test suites for \texttt{GameChannel} contain many well-written test cases where
game objects are explicitly specified by the ``\texttt{gameId}'' variable in each contract function.
Therefore, we can use such relations as a configuration to instruct how to automatically slice the transaction history according to the corresponding values of
``\texttt{gameId}'' to generate a set of independent game invocation sequences.

\subsection{Predicate Discovery from Dynamic Invariants}\label{sec:predicate}

The choice of predicates is crucial for computing good state abstractions.
In this paper, we use \emph{likely} pre- and post-conditions of contract functions as candidates.
Because of the blockchain transparency, we may decode the values of contract state variables and
user-provided function inputs, before and after each function invocation.
Then we statistically infer dynamic invariants for each function, which hold for all
observed invocations in the past transaction histories.
But since the transaction history may be limited, the inferred pre- and post-conditions are
\emph{likely} to hold, which is good enough to serve as predicate candidates.




More specifically, we define a \emph{predicate template} as ``$x \bowtie y$'', where $x \in X$ is a
parameter, and $y \in X \cup K$ is either a parameter or constant, and $\bowtie$ is an
operator from the set $\{\texttt{=}, \texttt{!=}, \texttt{>}, \texttt{<}, \texttt{<=},
\texttt{>=}\}$.
The template is instantiated on all successful transactions, which are not reverted during
executions, and the instances which always hold are kept as predicate candidates for either
function pre- or post-conditions.
The predicates defined over state variables are used in constructing the symbolic states $Q$ in EFSM.

The inference process is similar to how dynamic invariants are detected in Daikon-like systems~\cite{daikon} through a set of predefined invariant templates.
However, smart contracts are usually writting in Turing-complete programming languages such as Solidity, supporting complex data structures including array, mapping and custom struct.
Thus, we built our invariant inference on a tool called InvCon+~\cite{liu2024automated,Liu2022IAD} capable of invariant detection for smart contracts.

\subsection{Automata Construction}

The over-generalization of the inferred function pre-/post-conditions is the main difficulty for their direct use in mining high-level automata specifications.
To address this problem, we use a CEGAR-like approach to mine automata specifications with predicate abstraction.

\begin{figure*}[t]
\centering\small
\begin{prooftree}
	\AxiomC{$\bot$}
	\RightLabel{\textsc{Init}}
	\UnaryInfC{$\langle q_0 \gets \underset{x\in X_1}{\land} x=0, Q \gets \{q_0, \lnot q_0\},
	\Sigma, G \gets \{g_m\}_m, U \gets \{u_m\}_m, T \gets \emptyset \rangle$}
\end{prooftree}
\begin{prooftree}
	\AxiomC{$\langle q_0, Q, \Sigma, G, U, T \rangle$}
	\AxiomC{$\exists\; q_i, q_j \in Q \cdot (q_i \land g_m) \land (q_j \land u_m)$}
  \AxiomC{$\nexists\; t \cdot (q_i, g_m, e_m, u_m, q_j) \in T$}
	\RightLabel{\textsc{Construct}}
	\TrinaryInfC{$\langle q_0, Q, \Sigma, G, U,T \gets T \cup \{t\} \rangle$}
\end{prooftree}
\begin{prooftree}
	\AxiomC{\efsm: $\langle q_0, Q, \Sigma, G, U, T \rangle$}
	\AxiomC{\stackanchor{$\exists\; \pi_{n}: q_0t_1t_2 \cdots t_{n}q_{n} \in \efsm \; \exists\;
	\mathrm{concretize}(\pi_{n}) \in \history$}{$\exists\; \pi_{n+1}:  \pi_{n} \join
	t_{n+1}q_{n+1} \in \efsm \; \nexists\; \mathrm{concretize}(\pi_{n+1}) \in \history$}}
	\RightLabel{\textsc{RmPath}}
	\BinaryInfC{\efsm $\gets$ \textsc{SplitRemove}($q_n, t_{n+1}$, \efsm)}
\end{prooftree}
\caption{Specification mining rules.}
\label{fig:rules}
\end{figure*}

\begin{algorithm}[t]
	\caption{\textsc{SplitRemove}($q_n, t_{n+1}$, \efsm)}
	\begin{algorithmic}[1]\small
		\State Let $\langle q_0, Q, \Sigma, G, U, T \rangle = \efsm$
		\State Let $t_{n+1}\;=\;(q_n, g_m, e_m, u_m, q_{n+1}) \in T$  \Comment{a transition from state $q_n$ to $q_{n+1}$ by the invocation to function $e_m$ where $g_m$, $u_m$ are its precondition and post-condition, respectively.}
		\State $\hat{q_1} \; = \; q_n \land g_m$	 \label{line:split-1}
		\State $\hat{q_2} \;=\; q_n \land \lnot g_m$ \label{line:split-2}
		\If{$\mathrm{SAT}(\hat{q_1}) \land \mathrm{SAT}(\hat{q_2}) $}  \Comment{$q_n$ is splittable with $g_m$.} \label{line:sat}
		\State $Q \gets (Q \setminus q_n) \cup \{\hat{q_1}, \hat{q_2}\}$ \Comment{replace $q_n$ with
		two new states.} \label{line:sat-q}
		\State  Removes transitions starting or ending with $q_n$ in $T$ \label{line:sat-t}
		\Else \label{line:not-sat-begin}
		\If{$\nexists\; \pi_n' \cdot \pi_n' \in $\efsm $ \land \mathrm{concretize}(\pi_n'\join t_{n+1})
		\in$ \history} \label{line:notreachable}
		\State  $T \gets T \setminus t_{n+1}$  \Comment{remove unreachable transition $t_{n+1}$
	}
		\label{line:minust}
		\Else
		\State Let $S_{q_n | \pi_n'}$ be the set of all the concrete states of $q_n$ in the history \history, which are visited by the observed invocation sequences of $\pi_n'$. \label{line:other-start}
		\State $Q \gets (Q \setminus \{q_n\}) \cup \{ \mathrm{Pred}(S_{q_n | \pi_n'}), {q}_{n} \land
		\lnot \mathrm{Pred}(S_{q_n | \pi_n'})\}$ \label{line:other-split}
		\State Removes transitions starting or ending with $q_n$ in $T$ \label{line:other-remove}
		\EndIf
		\EndIf \label{line:not-sat-end}

		\State \Return $\langle q_0, Q, \Sigma, G, U, T \rangle$  \Comment{return the resulting
		automaton}	\label{line:return}
	\end{algorithmic}
	\label{algo:split-remove}
\end{algorithm}

\paragraph{\Cegar}
To mine a precise specification, the key is to compute a precise state abstraction $\alpha$, which
partitions the contract state.
The abstraction function $\alpha$ is implicitly computed following the paradigm of
\cegar~\cite{clarke2000counterexample}.
We define our specification mining algorithm by the three rules in \cref{fig:rules}.
Our algorithm takes as input past observations of concrete invocation sequences and
inferred function pre-/post-conditions.
When the algorithm terminates, it produces an \efsm containing no spurious states and transitions.

The \textsc{Init} rule initializes a preliminary extended finite state machine containing two states: $q_0$,
referring to $\underset{x\in X_1}{\land} x=0$ that all state variables are valued zero, and $\lnot q_0$ for the remaining cases.
The guard function $G$ and update function $U$ are directly instantiated by the inferred function pre- and postconditions, respectively.
Also, the transition relation set $T$ is initialized to be empty.
Then, we apply the \textsc{Construct} rule to add theoretically feasible state transitions to the automaton.
A state transition is theoretically feasible if and only if it satisfies the logical conjunction of symbolic states and function preconditions or post-conditions.
The resulting automaton could be over-generalized such that it includes spurious state transition paths.
Therefore, we need to apply the \textsc{RmPath} rule, following \cref{algo:split-remove} to
eliminate those spurious state transition paths that are not supported in the concrete observations.
\Cref{algo:split-remove} rules out spurious paths by either state splitting or transition removal.
These rules would be applied many times according to a fair scheduling.
When the algorithm terminates, it produces an \efsmtext, containing no spurious states or transitions.
The illustration of a running example and the fair scheduling and its proof are available at \website.

\paragraph{Loop transitions}
\label{subsec:genearalization}
The resulting automaton does not allow loop transitions according to the \textsc{RmPath} rule.
However, this kind of automaton may not be precise and useful contract specifications.
Because many smart contracts have behavior cycles, it is preferred to have loop transitions in the
resulting automaton.
Therefore, we limit the range of path selection when applying \textsc{RmPath}, i.\,e., a loop
transition can only be covered once in any selected path.
For example, a state transition path $q-\mathit{Event}_a-q-\mathit{Event}_a-q$ is not under our consideration when allowing loop transitions.
With this minor modification to the \textsc{RmPath} rule, the resulting automaton allows loop
transitions so that it may express cycles.

%% file: sections/Evaluation.tex
\section{Implementation and Evaluation}\label{sec:eval}
\subsection{Implementation}
We implement trace slice approach and specification mining algorithm as a tool named \tool, written in around 3K lines of Python code.
Specifically, we apply our trace slicing approach
to retrieve independent user action traces from transaction histories according to the given trace slice configurations, and then we invoke InvCon+ to produce corresponding likely invariants.
Based on these sequences and likely invariants, we are able to perform specification mining for smart contracts.
Additionally, our algorithm relaxes the \textsc{RmPath} rule to allow loops in the contract specifications for
better generality (see~\cref{subsec:genearalization}).
We used the Z3 SMT solver~\cite{moura2008z3} for discharging satisfiability queries.


We generate function-level invariants for smart contracts from the past transaction histories and
filter the generated invariants to keep those expressing parameter relations
(see~\cref{sec:predicate}).
These invariants serve as the parameter predicates that we use for automata construction (see~\cref{sec:algorithm}).

Through experiments, we evaluated \tool to answer the following three research questions:
\begin{itemize}[leftmargin=*]
\item \textbf{RQ1}: How effectively does \tool mine smart contract specifications  {compared with the state-of-the-arts}?
\item \textbf{RQ2}: How effectively does \tool mine automata from real-world DApp smart contracts, and with these automata, how is symbolic analysis for smart contracts enhanced?
\item \textbf{RQ3}: What are the implications for DApp developers?

\end{itemize}

\begin{table}[t]\centering
	\caption{The Azure smart contract benchmark.}\label{tab: }
	\scriptsize
	\begin{tabular}{llrr}\toprule
		\multirow{2}{*}{Contract} &\multirow{2}{*}{Description} &\multicolumn{2}{c}{Formal Specifications} \\\cmidrule{3-4}
		& &\# States &\# Transitions \\\midrule
		AssetTransfer &Selling high-value assets &11 &32 \\
		BasicProvenance &Keeping record of ownership &4 &4 \\
		BazaarItemListing &Selling items &4 &5 \\
		DefectCompCounter &Product counting &3 &2 \\
		DigitalLocker &Sharing digital files &7 &12 \\
		FreqFlyerRewards &Calculating flyer rewards &3 &3 \\
		HelloBlockchain &Request and response &3 &3 \\
		PingPongGame &Two-player games &4 &2 \\
		RefrigTransport &IoT monitoring &5 &8 \\
		RoomThermostat &Thermostat installation and use &3 &4 \\
		SimpleMarketplace &Owner and buyer transactions &4 &4 \\
		\midrule
		Average & &4.64 &7.18 \\
		\bottomrule
	\end{tabular}
\end{table}

\subsection{Methodology}
To answer RQ1, we evaluate \tool on parametric-free smart contracts from a well-studied benchmark used for Azure enterprise blockchain, where none of these contracts have index-related data structures so we do not perform trace slicing on their transactions.
This benchmark includes 11 smart contracts exhibiting stateful behaviors, ranging over supply chain management, digital control, virtual games, etc.
Each of these contracts is properly documented, and their specifications have been well formalized and examined by the previous work.
Such ground truth specifications are deemed as the reference models in our evaluation.
Because \tool aims to dynamically infer specification models from past contract executions,
we produce 10,000 transactions per contract using random test case generation.
In detail, we deploy every contract 100 times to our testnet.
Each contract instance is tested using 100 randomly generated transactions, which finally produce a trace, namely a sequence of contract executions.
Subsequently, we perform \tool on these contract traces to mine contract specification models.



To answer RQ2,
we evaluate \tool on real-world parametric smart contracts running on Ethereum.
We selected six popular Ethereum DApp smart contracts as shown
in \cref{tab:dapp}.
We selected them from the Top-10 DApps covering different application
domains~\cite{dappradar}, such as decentralized gaming, gambling, non-fungible token (NFT) usage,
and an exchange market.
For example, the DApp \texttt{SuperRare} has a total trading volume up to 557 million dollars
contributed by more than 10,000 users in nearly 100,000 transactions~\cite{SuperRare};
and MoonCatRescue has a total trading volume up to 73 million dollars involving more than 11,000 users~\cite{MoonCatRescue}.
These DApps have been deployed and running for a long period, since as early as 2017, and their
past transaction data can be downloaded from Ethereum.
Most of these DApps (except 0xfair) maintain some form of design documentation on their websites or
GitHub repositories; some also provide formal specifications, such as Dicether~\cite{dicether}.
In addition, well-organized DApp projects, such as the studied ones, maintain test suites that
exercise the core functionalities of the contracts with reasonable coverage.
With these artifacts, we are able to construct ground models manually for DApp contracts.
We collected their contract code and transaction data from Etherscan~\cite{Etherscan} and Ethereum archive node hosted by QuickNode~\cite{QuickNode}.
Particularly, the number of transactions used for specification mining is also capped at 10,000 for all DApp smart contracts.

\paragraph{Evaluation Metrics}
To evaluate \tool, we use the accuracy metric recommended in~\cite{lo2006quark} for automata
specification mining evaluation.
The accuracy metric measures the similarity between the mined automata specification and the ground
truth, considering both precision and recall.
Precision is defined as the percentage of sequences generated by the mined automata that are
accepted by the ground truth, while recall is the percentage of sequences generated by the ground
truth that are accepted by the mined automata.
Following \cite{le2018deep}, we use the $F_1$-score to measure the overall accuracy, which is
defined as:
$
F_1 = \frac{2 \times \mathit{Precision} \times \mathit{Recall}}{\mathit{Precision} +
	\mathit{Recall}}.
$
Since automata may have infinite sequences when they have loop transitions, to obtain accurate
precision, recall and $F_1$-score, we follow the similar strategies used in
previous works~\cite{le2018deep, le2015synergizing, lo2012learning} to generate the sequences.
We set the maximum number of generated sentences to 10,000 with minimum coverage of each transition to be 20 in the generated traces~\cite{le2018deep} and restrict the length of the traces to twice the number of transitions~\cite{krka2014automatic} in the ground-truth models that have been formalized by the Azure benchmark or manually constructed by ourselves.
In addition, for RQ1, we divide the transaction data into a training and a test set, where we mine the model from contract executions in the training set.
We use another accuracy metric, denoted as Acc, to measure how many percentages of contract executions in testing set are accepted by the mined model.


\subsection{Experiment Setup}
All experiments were conducted on an Ubuntu 20.04.1 LTS desktop equipped with an Intel Core i7 $16$-core processor and $32$~GB of memory.
The ground truth, benchmark contracts, and raw results are available at: \website.

\begin{table*}[t]\centering
	\caption{Experiement results on the Azure benchmark.}\label{tab:comparison}
	\small
	\resizebox{\textwidth}{!}{
\begin{tabular}{l@{}rrr|rrr|rrr|rrr|rrr|rrr}\toprule
	\multirow{2}{*}{Contract} &\multicolumn{3}{c|}{\onetail} &\multicolumn{3}{c|}{\twotail} &\multicolumn{3}{c|}{\onesekt} &\multicolumn{3}{c|}{\twosekt} &\multicolumn{3}{c|}{\contractor} &\multicolumn{3}{c}{\tool} \\\cmidrule{2-19}
	&\# States &$F_1$ &Acc &\# States &$F_1$ &Acc &\# States &$F_1$ &Acc &\# States &$F_1$ &Acc &\# States &$F_1$ &Acc &\# States &$F_1$ &Acc \\\midrule
AssetTransfer &24 &0.52 &0.93 &40 &0.47 &0.77 &24 &0.52 &0.93 &40 &0.47 &0.77 &13 &0.2 &1 &13 &0.34 &0.97 \\
BasicProvenance &4 &0.72 &1 &6 &0.67 &1 &4 &0.67 &1 &6 &0.7 &1 &3 &0.63 &1 &3 &0.8 &1 \\
BazaarItemListing &9 &0.94 &1 &94 &0.97 &0.84 &9 &0.94 &1 &83 &0.98 &0.87 &3 &0.89 &1 &3 &1 &1 \\
DefectCompCounter &3 &1 &1 &3 &1 &1 &3 &1 &1 &3 &1 &1 &3 &1 &1 &3 &1 &1 \\
DigitalLocker &18 &0.57 &0.95 &29 &0.34 &0.94 &18 &0.57 &0.95 &29 &0.34 &0.94 &9 &0.95 &1 &10 &0.87 &1 \\
FreqFlyerRewards &3 &1 &1 &5 &1 &1 &3 &1 &1 &5 &1 &1 &2 &1 &1 &2 &1 &1 \\
HelloBlockchain &4 &1 &1 &5 &1 &1 &4 &1 &1 &5 &1 &1 &3 &1 &1 &3 &1 &1 \\
PingPongGame &4 &0.77 &1 &4 &0.75 &1 &4 &0.77 &1 &4 &0.75 &1 &5 &0.51 &1 &4 &0.77 &1 \\
RefrigTransport &6 &0.7 &1 &8 &0.68 &1 &6 &0.7 &1 &8 &0.69 &1 &5 &0.43 &1 &5 &0.69 &1 \\
RoomThermostat &5 &0.88 &1 &9 &0.88 &1 &5 &0.88 &1 &9 &0.88 &1 &5 &1 &1 &6 &1 &1 \\
SimpleMarketplace &5 &1 &1 &6 &1 &1 &5 &1 &1 &6 &1 &1 &4 &1 &1 &5 &1 &1 \\
\midrule
Average &7.73 &0.83 &0.99 &19.00 &0.80 &0.96 &7.73 &0.82 &0.99 &18.00 &0.80 &0.96 &5.00 &0.78 &1.00 &5.18 &0.86 &1.00 \\
\bottomrule
\end{tabular}
}
\end{table*}

\subsection{RQ1. Effectiveness of \tool}
To answer RQ1, we compared \tool with five baseline approaches.
K-tail~\cite{biermann1972synthesis} learns an automaton from prefix trees of traces
by merging nodes with the same `tail' of length $k$.
We evaluated its two settings, \onetail when $k=1$ and \twotail when $k=2$.
SEKT~\cite{krka2014automatic} is a type of state-enhanced k-tail, which extends k-tail using
program state information inferred from the full set of observed executions.
We also evaluated its two settings, \onesekt when $k=1$ and \twosekt when $k=2$.
\contractor~\cite{de2010automated, krka2014automatic} creates finite state machine models
exclusively based on program invariants inferred from the observed executions.
To the best of our knowledge, there exists only one other approach to mine state machine models
from smart contract executions, by Guth et al.~\cite{guth2018specification}.
However, their tool was not available for comparison at the time of writing.
We will discuss this related work and compare with it in~\cref{sec:related}.

The original benchmark contracts do not always satisfy the specifications that come with them,
which has also been revealed by a previous study~\cite{wang2019formal}.
For a fair comparison with the other approaches, we manually repaired these issues and also
reported them to the developer~\cite{issue278,issue279, issue280, issue281}.
For instance, \code{SimpleMarketplace} is a contract application that implements a workflow for a
simple transaction between an owner and a buyer in a marketplace.
\code{SimpleMarketplace} has an \code{AcceptOffer} function to allow owner to accept the offer made
by buyers.
However, \code{AcceptOffer} even succeeds when there is no offer placed, thus violating its formal
specification~\cite{issue281}.

\paragraph{Evaluation results}
\Cref{tab:comparison} provides a detailed overview of the comparative performance of various tools, including our developed tool \tool, in the domain of smart contract specification mining. Each row corresponds to a specific smart contract, with columns showcasing essential metrics such as the number of state machine models generated (\# States), the F-score ($F_1$), and the accuracy (Acc). The evaluated tools, denoted as \onetail, \twotail, \onesekt, \twosekt, \contractor, and \tool, allow for a comprehensive analysis of their capabilities in extracting and representing contract specifications. The variety of contracts considered, ranging from \code{AssetTransfer} to \code{SimpleMarketplace}, ensures a diverse and thorough assessment of each tool's performance across different use cases.
\revise{We do not compare with grammar inference and deep learning techniques since our preliminary experiments with the minimal-description-length grammar inference by LearnLib~\cite{learnlib} indicate that grammar inference tends to overgeneralize, having very poor precision, while deep learning techniques demand a large volume of training data that is difficult to collect from real-world transactions.}   

Upon closer examination of the data, it is evident that \tool consistently exhibits competitive performance metrics, followed by \contractor. Notably, in the AssetTransfer contract, \tool outperforms \contractor by generating a state machine model with 13 states, resulting an higher $F_1$ score of 0.34 with neglectable loss of precision. Across all contracts, \tool maintains an average of 5.18 states per model, an impressive $F_1$ score of 0.86, and nearly perfect accuracy (1.00). These results highlight the efficacy of \tool in accurately capturing the intricacies of smart contract behavior. The tool's robust performance, in terms of model compactness, accuracy and $F_1$ score, distinguishes it from other baseline approaches, emphasizing its potential as a reliable solution for specification mining tasks.

In summary, \tool emerges as a promising tool for smart contract specification mining, achieving
good precision, recall, and accuracy for around three minutes per contract.
The presented results demonstrate its consistent ability to generate accurate state machine models across a diverse set of contracts.
The high average $F_1$ score and accuracy substantiates effectiveness in capturing the intended behavior of smart contracts.
These results establish \tool as a valuable resource for researchers and practitioners in search of a dependable and adaptable tool for real-world smart contract analysis.
\begin{table*}[!htp]\centering
	\caption{Experiment results on real-world DApp contracts.}\label{tab:dapp}
	\scriptsize
	\resizebox{\textwidth}{!}{
	\begin{tabular}{lrrrrrrrrrrrrrrrr}\toprule
		&\multirow{3}{*}{Description} &\multicolumn{2}{c}{\multirow{1}{*}{Mined}} &\multicolumn{6}{c}{Opcode Coverage} &\multicolumn{6}{c}{Number of Issues} \\\cmidrule{5-16}
		& & \multicolumn{2}{c}{\multirow{1}{*}{Specifications}} &\multicolumn{2}{c}{Mythril-Random} &\multicolumn{2}{c}{Mythril-SMCon} &\multicolumn{2}{c}{Statistics} &\multicolumn{2}{c}{Mythril-Random} &\multicolumn{2}{c}{Mythril-SMCon} &\multicolumn{2}{c}{Statistics} \\\cmidrule{3-16}
		& & \#State &\#Tran. &Avg. &Var. &Avg. &Var. &p-value &$\hat{A_{12}}$ &Avg. &Var. &Avg. &Var. &p-value &$\hat{A_{12}}$ \\\midrule
		CyptoKitties &kitty auction &2 &5 &68.20\% &0.0016 &70.21\% &0.0214 &0.3785 &0.6667 &4.5 &0.3 &6.5000 &7.5000 &0.0677 &0.6667 \\
		CryptoPunks &punk market &20 &81 &23.72\% &0.0003 &35.72\% &0.0162 &0.0345 &0.8333 &1 &0.4 &1.3333 &0.6667 &0.2243 &0.6389 \\
		SupeRare &art market &15 &88 &24.58\% &0 &28.12\% &0.0130 &0.2415 &0.3333 &0 &0 &0.3333 &0.2667 &0.0873 &0.6667 \\
		MoonCatRescue &cat adoption &18 &70 &18.88\% &0.0003 &38.76\% &0.0028 &0.0001 &1 &1 &0 &1 &0 &NA &0.5000 \\
		0xfair &RPS game &4 &5 &46.11\% &0.0003 &50.09\% &0.0000 &0.0011 &1 &4 &0 &5 &0 &NA &1 \\
		Dicether &bet game &8 &16 &50.83\% &0.0022 &32.99\% &0 &0.0001 &0 &4 &0.8 &3 &0 &0.0204 &0.1667 \\
		\midrule
		Average & & & &38.72\% &0.0008 &42.65\% &0.0089 &0.1093 &0.6389 &2.4167 &0.25 &2.8611 &1.4056 &0.1000 &0.6065 \\
		\bottomrule
		\end{tabular}
	}
\end{table*}

\begin{figure}[t]
	\includegraphics[width=\columnwidth, trim=0 0 0 0]{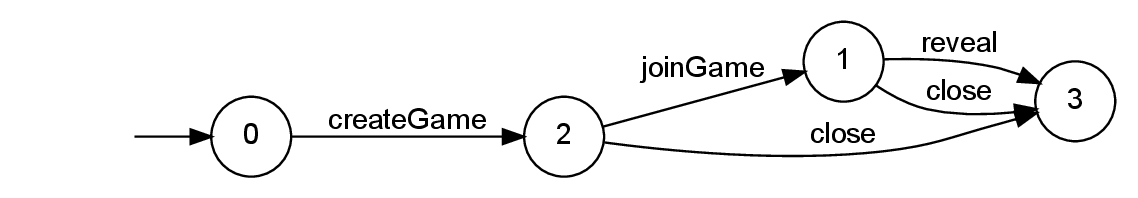}
	\caption{The mined automaton for 0xfair. Note that we exclude parametric bindings for simplicity.}\label{fig:0xfair}
\end{figure}

\begin{figure}[t]
    \centering
	\scriptsize
    \begin{subfigure}{0.32\columnwidth}
        \centering
		\includegraphics[width=\linewidth]{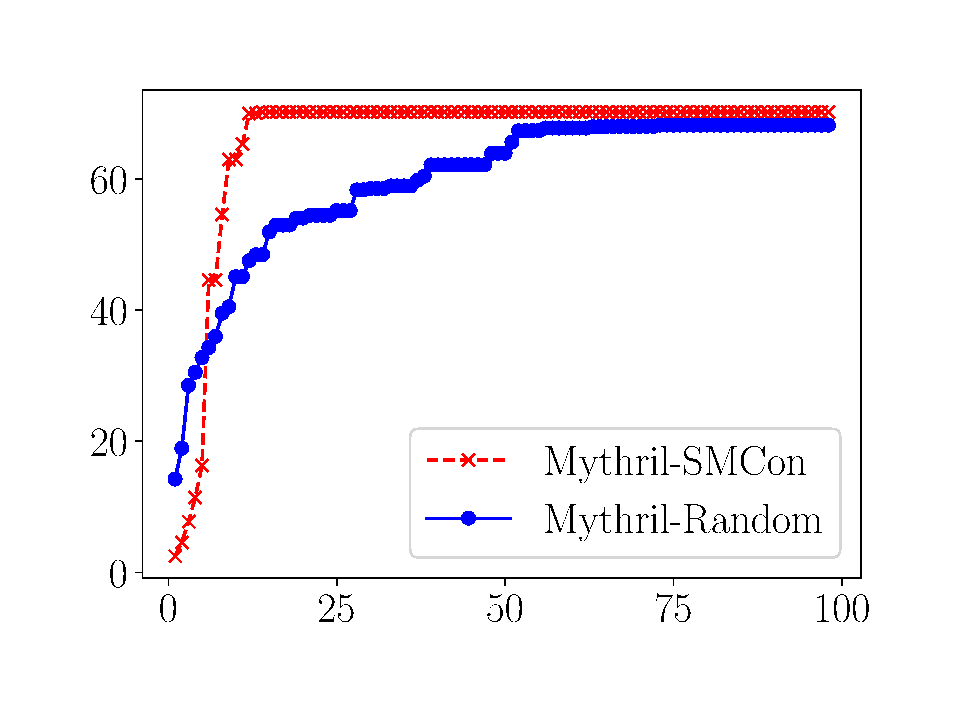}
		\vspace{-20pt}
        \caption{\scriptsize CryptoKitties}
    \end{subfigure}
    \hfill
    \begin{subfigure}{0.32\columnwidth}
        \centering
		\includegraphics[width=\linewidth]{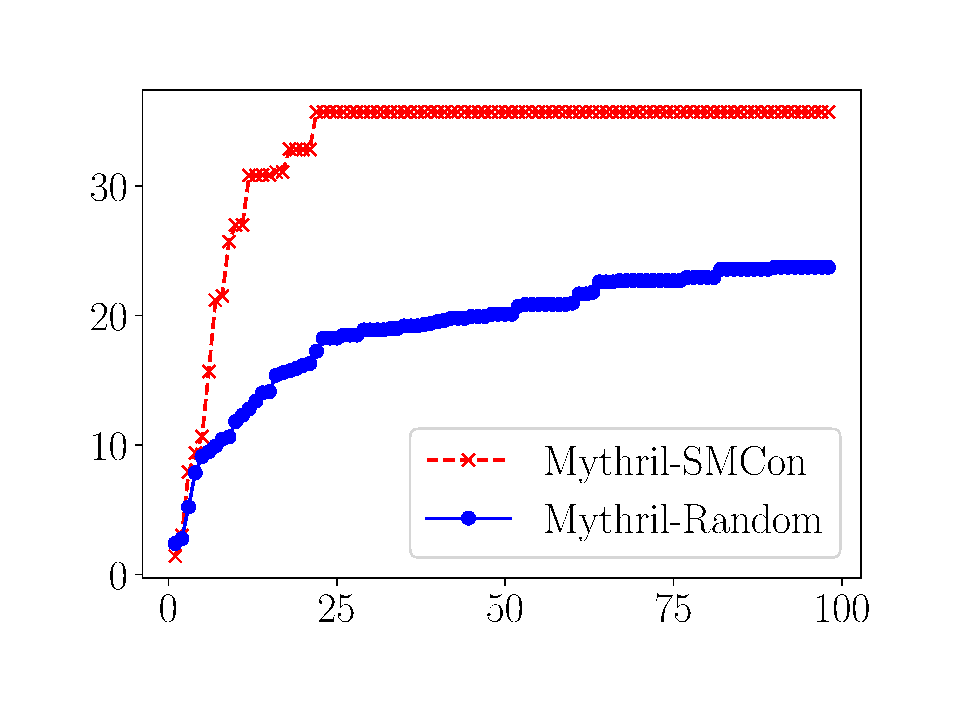}
        \vspace{-20pt}
		\caption{\scriptsize CryptoPunks}
    \end{subfigure}
	\hfill
    \begin{subfigure}{0.32\columnwidth}
        \centering
		\includegraphics[width=\linewidth]{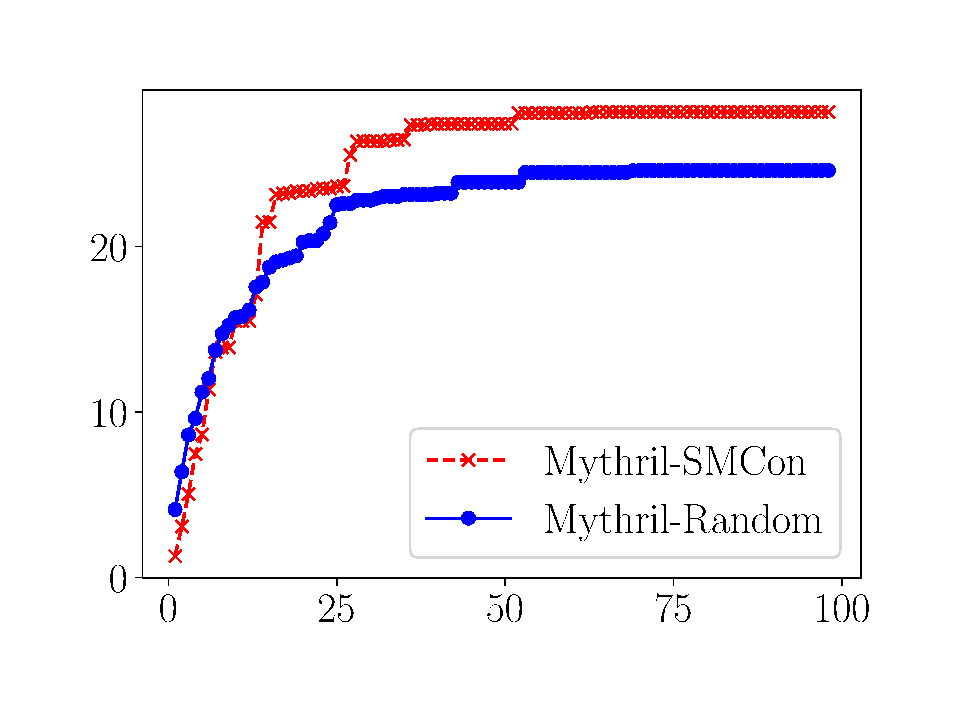}
		\vspace{-20pt}
        \caption{\scriptsize SupeRare}
    \end{subfigure}
	\hfill
	\begin{subfigure}{0.32\columnwidth}
        \centering
		\scriptsize
		\includegraphics[width=\linewidth]{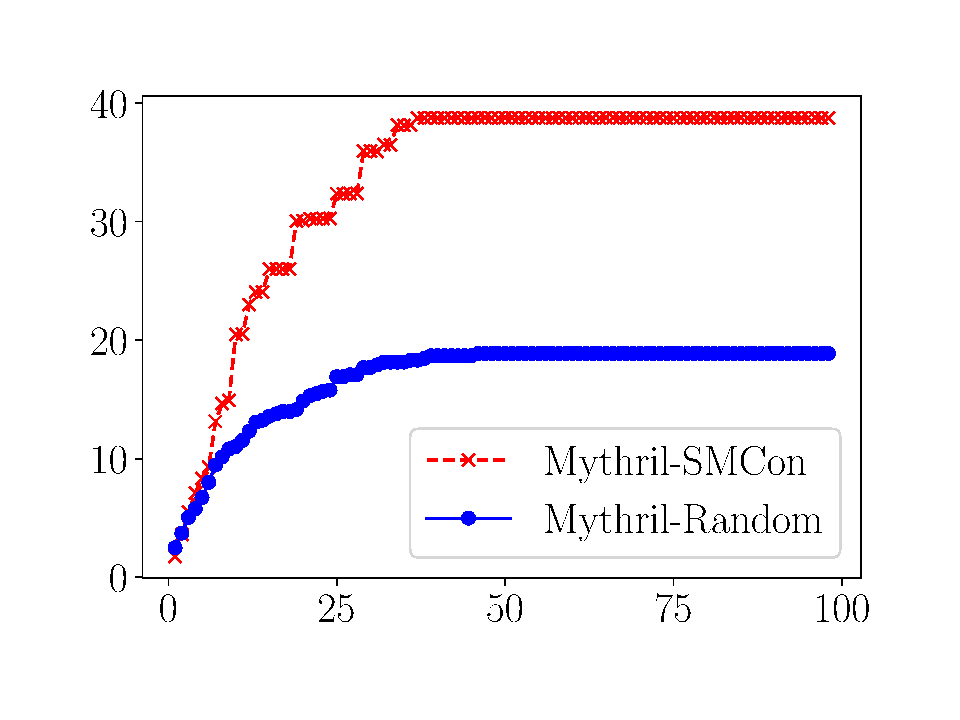}
		\vspace{-20pt}
        \caption{\scriptsize MoonCatRescue}
    \end{subfigure}
    \hfill
    \begin{subfigure}{0.32\columnwidth}
        \centering
		\scriptsize
		\includegraphics[width=\linewidth]{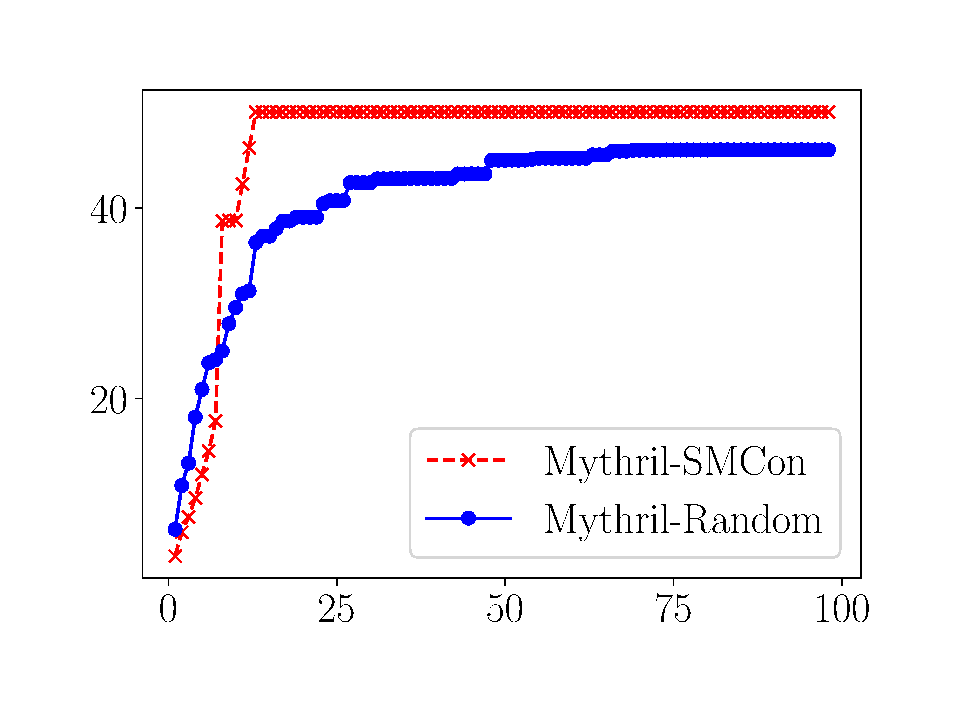}
		\vspace{-20pt}
        \caption{\scriptsize 0xfair}
    \end{subfigure}
	\hfill
    \begin{subfigure}{0.32\columnwidth}
        \centering
		\includegraphics[width=\linewidth]{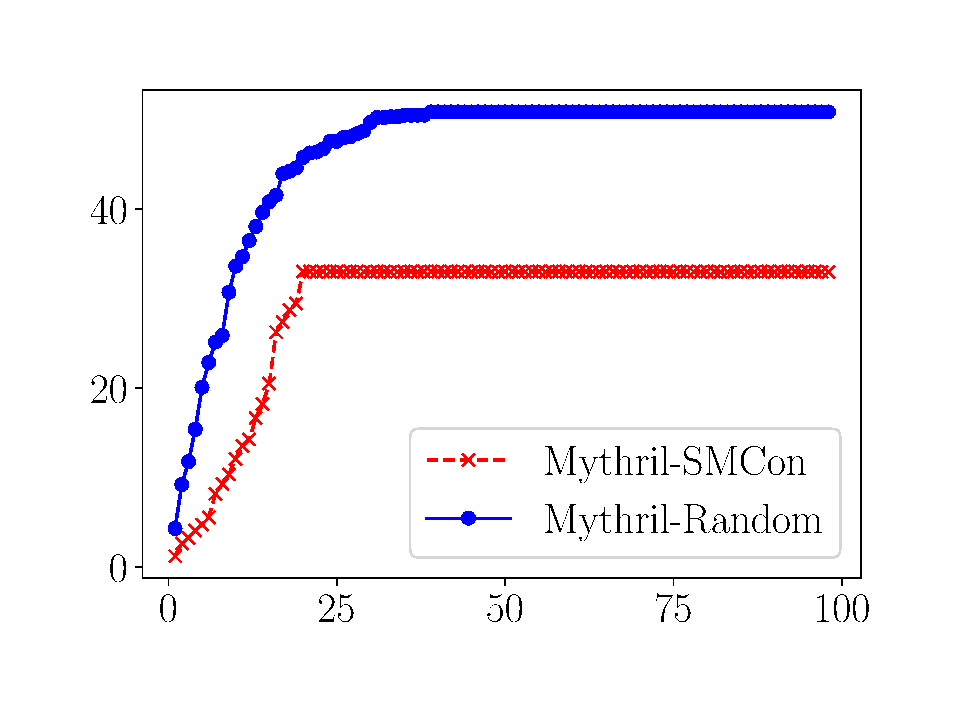}
		\vspace{-20pt}
        \caption{\scriptsize Dicether}
    \end{subfigure}
    \caption{Opcode coverage achieved with the number of function call sequences used. The x-axis and y-axis indicate the number of function call sequence and the percentage of opcode coverage, respectively.}
    \label{fig:efficiency}
\end{figure}

\subsection{RQ2. Experiment Results on Real-world Smart Contracts}
\Cref{tab:dapp} illustrates the {automata mining} results of \tool on {six} real-world DApp contracts.
The model complexity of specifications mined varies a lot.
\code{CryptoKitties} has the simplest model with two states and three transitions.
The model can be interpreted as a regular language ``$(\mathit{createAuction} \rightarrow
\mathit{bid}\; |\; \mathit{cancelAuction})*$'', where each active auction accepts only one bid.
\Cref{fig:0xfair} shows the specifications mined for
\code{0xfair},\footnote{\url{https://etherscan.io/address/0xa8f9c7ff9f605f401bde6659fd18d9a0d0a802c5}}
which perfectly articulates the usage scenarios of a \textbf{R}ock-\textbf{P}aper-\textbf{S}cissor
game.
\code{0xfair} employs a seal mechanism to achieve fairness where nobody can cheat on others.
First, the creator encrypts his choice and publicizes the choice proof, namely, the corresponding
cryptographic signature, when creating a game via \code{createGame}.
Naturally, the second player joins this game with an explicit choice via \code{joinGame}.
Finally, the creator reveals his choice by decrypting the choice with the secret key, which is used
to determine the game winner.
In addition, a game should be closed when it expires, because no other players join or the creator
fails to reveal the choice.

The remaining four DApps have more complex models, which we omit due to limited space.
We manage to assess these mined models through their existing test suites.
We first re-ran all the test cases for each DApp and found that many test cases failed.
For example, \code{MooncatRescue} has 1,119 test cases, with 993 passing and 126 failing {mostly due to VM error messages slightly unmatched with the expected.}
Next, we were able to construct their ground truth specifications manually, where two of the authors spent two hours per smart contract individually.
In particular, we ensured that the tested behaviors must be included in the ground-truth specifications and added any additional untested behaviors clearly documented.
Because the transaction data used is collected from an uncontrolled blockchain environment, the
diversity of historical usage behaviors have a considerable impact on the mined model.
Our study shows that GameChannel achieved the highest precision of 93.1\,\% and recall of 97.9\,\%, respectively, while SupeRare scored the second highest recall of 96.1\,\% for SupeRare, followed by CryptoPunks achieving 80.7\,\%, and MoonCatRescue achieved the second highest precision of 85.6\,\%.
Therefore, we believe our mined models for real-world DApp contracts shall capture widely-used high-level program specifications, which can be used to enhance DApp development, e.g., uncovering issues of DApp document and test suites.


We study the effectiveness of the resulting automata for symbolic analysis of smart contracts.
Unlike testing, symbolic analysis often yields a more comprehensive security report by effectively exploring multiple program paths at once.
Nevertheless, symbolic analysis may face path explosion problem, which largely affects its performance.
In \Cref{tab:dapp}, we compare two usages of the state-of-the-art industrial symbolic analysis tool
named Mythril~\cite{mythril}---by providing randomly generated function call sequences, i.e.,
Mythril-Random, and by providing function call sequences generated from automata specifications
minded by \tool, i.e., Mythril-\tool.
For each contract, we cap the length of function call sequences to be five, the time budget to be
one hour while the timeout of symbolic execution of a function call sequence is set to 10 minutes.
For reliable comparison, we repeated such symbolic analysis process six times per contract.
Notice, for contract functions absent in our mined automata, we perform random selection and insert
the selected ones into the function call sequence generated by the model.

Smart contracts are compiled into opcodes executable on the Ethereum Virtual Machine.
\Cref{tab:dapp} shows the opcode coverage achieved and the number of issues reported by Mythril-Random and Mythril-\tool.
At first glance, most of the code coverage statistics seem low.
This is partly because we only test the public contract functions that could alter program states
while leaving untested the other view functions that only access program states.
The timeout setting also has an impact on this, and we will explain it later.
Overall, Mythril-\tool achieves 42.65\,\% code coverage and finds around 3 issues per contract, which is more than what Mythril-Random achieves.
Note, the issues reported by Mythril are often considered as warnings for developers to check,
which may not always reflect real vulnerabilities.
In detail, Mythril-\tool outperforms Mythril-Random in all cases except Dicether.
Moreover, Mythril-\tool is proved more likely to explore new program paths, since Mythril-\tool
displays a larger variance of the code coverage than Mythril-Random except 0xfair and Dicether.
Our further investigation shows that Dicether has two preparation functions to execute before any game-related operations, which are not included in the mined automata, but such problem can be mitigated using dependency analysis~\cite{wang2020oracle}.
For 0xfair, the program path constraints are too complicated to solve within the given timeout, which will be illustrated in~\cref{fig:speedup}.
We also perform statistics analysis using Mann Whitney U-test to show the significance level of the experiment result and Vargha and Delaney's A12 statistical test to determine the extent to which Mythril-\tool outperforms Mythril-Random.
The results in~\cref{tab:dapp} indicate that in terms of code coverage or number of issues reported, Mythril-\tool usually performs better than Mythril-Random in 4 out of 6 cases, with resulting $\hat{A_{12}}$ scores exceeding 0.6 and p-values being smaller than or close to a significance level of 0.05.
We also delve into how \tool promote efficiency of symbolic analysis in achieving good opcode
coverage with much less function call sequences.
As shown in~\cref{fig:efficiency}, except Dicether, \tool helps Mythril reach a higher opcode
coverage with less number of function call sequences compared to its random counterpart,
highlighting the usefulness of the specifications mined.
\revise{We acknowledge that there are many studies that improve fuzzing effectiveness by incorporating valuable feedback information from static analysis~\cite{grieco2020echidna} and dynamic analysis~\cite{wang2020oracle, shou2023ityfuzz}. The high-level behavior automata mined by~\tool align with this field and complement these existing fuzzing tools.}

To speed up symbolic analysis, we could also enforce the trace slice setting by fixing trade
session parameter, e.g., gameId, to a constant because most trade sessions are homogeneous,
non-interleaving and symbolic analysis of a trade session should suffice.
To investigate this impact, we sampled the function call sequences derived from the mined automata
of the DApp contracts where each sequence represents a particular trade scenario.
\Cref{fig:speedup} draws the overall time consumptions for default and trace slice setting, where for each setting, we symbolically execute each sequence five times.
Trace slice setting takes smaller time for all cases except CryptoKitties, where MoonCatRescue has
56\% speedup, followed by 0xfair's 36\%.

In summary, the automata infered by \tool about high-level program behaviors is critical to reduce
the burden of symbolic analysis for complicated smart contracts, and it can complement existing
speedup techniques, such as predicting unsatisfiable symbolic path with machine learning
models~\cite{yang2024exploring}.

\begin{figure}
	\centering
	\includegraphics[width=.9\columnwidth]{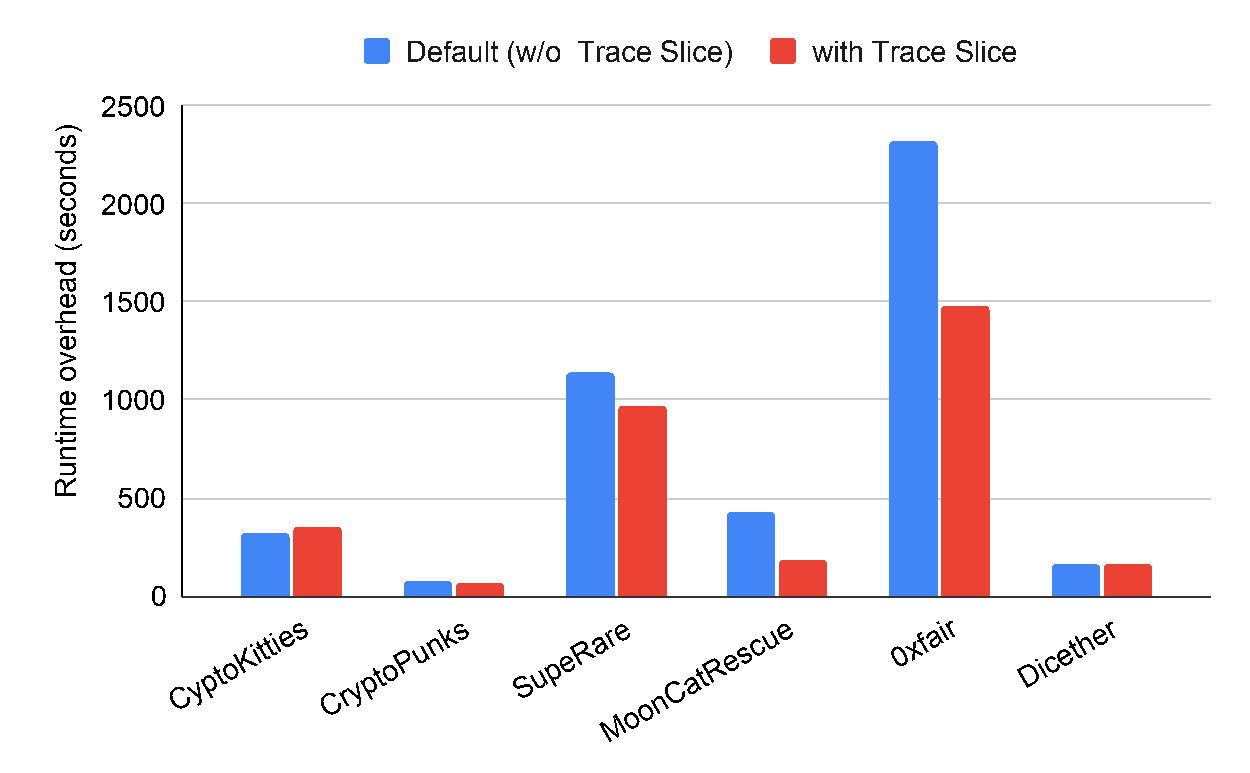}
	\caption{Time consumption of symbolic analysis with or without enforcing trace slicing.}
	\label{fig:speedup}
\end{figure}

\subsection{RQ3. Implications in DApp Development}
\paragraph{Outdated Documentation}
Management of documentation and ensuring its consistency with contract implementation is often
labor-intensive.
For instance, \code{Dicether} is a gambling game running on Ethereum, first launched in 2018.
We investigated and collected eight contract versions of Dicether so far from DAppRadar~\cite{dappradar} and Etherscan~\cite{Etherscan} where  maintenance occurs the most in its first year and each contract version lived for about two months on average.
By differentiating between these contract versions, we found some contract maintainence performs
only routine tasks, e.g., minimal patching for security and reliability considerations.
In contrast, some maintenance updates introduce substantial changes to business logics.
When comparing its first contract
version\footnote{\url{https://etherscan.io/address/0xc95d227a1cf92b6fd156265aa8a3ca7c7de0f28e}}
with the seventh contract
version,\footnote{\url{https://etherscan.io/address/0xaec1f783b29aab2727d7c374aa55483fe299fefa}}
we noticed function renaming changes, e.g., \code{playerCancelActiveGame} to
\code{userCancelActiveGame}.
Additionally, for the original contract version, the server user or normal player cannot perform any operation until a created game session is accepted.
However, this business logic was removed in the seventh contract version.
Yet, the only formal documentation~\cite{dicether} of the \code{Dicether} design is outdated and
can no longer reflect the program behaviors of the recently used contract versions.
To summarize, we believe \tool could mine high-level automata for evolving smart contracts that can
facilitate developers to track new changes easily and maintain high-quality documentation.

\paragraph{Test Suite Bias}
Developing test suites for smart contracts is non-trivial since developers usually have little knowledge of how smart contracts are used after contract deployment to blockchains, thus crafting test cases for functions that are rarely used could waste human efforts and missing test cases for functions that are heavily used could leave a room for security risks.
For example,
\code{CryptoPunks} is one of the earliest examples of using Non-Fungibale Tokens (NFTs) on Ethereum, which inspired the ERC-721 standard to some extent.
CryptoPunks has 10,000 unique collectible characters called punks, with proof of ownership stored on Ethereum~\cite{cryptopunks}.
To start with, function
\texttt{getPunk} or \texttt{setInitialOwners} of CryptoPunks is called to assign punks to users.
Users can transfer the ownership of a punk by calling \texttt{transferPunk}.
Users can make a bid to a punk via \texttt{enterBidForPunk}.
CryptoPunks has a set of test suites in its GitHub repository covering seven use scenarios such as
setting the initial owner(s) of punks or opening a sale for punks~\cite{cryptopunkgittest}.
However, in the existing test suites, there is only one test case for \texttt{setInitialOwners},
while the other test cases all focus on \texttt{setInitialOwner}.
An interesting observation is that, based on transaction histories, the contract manager always
use \texttt{setInitialOwners} to initialize a batch of punks for a group of owners instead of
\texttt{setInitialOwner} for individual assignment.
Our mined automata highlights this disproportional focus on rarely used functions,
while inadequate tests were written for more frequently used functionalities.
For example, it may be expected to test \texttt{setInitialOwners} whether a punk assigned to one owner can be wrongly overwritten by a succeeding owner in the same group, which is indeed not enforced in the current contract implementation.

\paragraph{Threats to validity}
An internal threat is potential errors in the manually derived ground truth for DApp contract specifications. To mitigate this, we collected well-documented smart contracts from popular DApp projects and re-ran their test suites. Additionally, our tool implementation and experimental scripts might contain bugs. Two authors closely collaborated on the tool and reviewed the code regularly. We also checked for outliers in the results, uncovering and fixing a few bugs.
Externally, our findings may not generalize to all DApp smart contracts. To address this, we selected representative DApps from various application domains.

%% file: sections/RelatedWork.tex
\section{Related Work}
\label{sec:related}
\paragraph{Smart Contract Specification Mining}
Several tools have been developed for mining smart contract specifications, which can be categorized into low-level functional specifications~\cite{tan2022soltype, liu2022learning, wang2024smartinv, Liu2022IAD, liu2024automated} and high-level behavioral specifications~\cite{guth2018specification}.
SolType~\cite{tan2022soltype} focuses on Solidity smart contracts, allowing developers to add refinement type annotations for static analysis of arithmetic operations. While it effectively detects issues like integer overflows, it is limited to contract-level arithmetic invariants. Cider~\cite{liu2022learning} extends SolType by using deep reinforcement learning to infer contract invariants, but these remain unverified. SmartInv~\cite{wang2024smartinv} takes a multimodal learning approach to infer invariants that identify hard-to-detect bugs. InvCon~\cite{Liu2022IAD} and InvCon+~\cite{liu2024automated} infer invariants from blockchain transactions, with VeriSol~\cite{wang2019formal} verifying their correctness. However, these tools only infer invariants at function boundaries and do not capture higher-level state transitions.

Guth et al.~\cite{guth2018specification} mine specifications by slicing transaction histories into independent sequences and constructing a finite state machine (FSM) based on data dependencies. Our approach differs in two key ways: (1) we use test suite-based contract interaction patterns for more precise slicing, and (2) we mine extended finite state machines (EFSM), which are more expressive than traditional FSMs.

\paragraph{Automata Mining}
Automata mining has a rich history~\cite{de2010grammatical, le2018deep, aarts2012automata, krka2014automatic, walkinshaw2016inferring, biermann1972synthesis, lorenzoli2008automatic, beschastnikh2011leveraging}. Traditional approaches, such as grammar inference~\cite{gold1967language} and counterexample-guided abstraction refinement (CEGAR)~\cite{clarke2000counterexample}, have been applied to learn behavioral models of systems.
Aarts et al.~\cite{aarts2012automata} built on the L* algorithm~\cite{angluin1987learning} to generate restricted EFSMs from dynamic execution traces. RPNI-MDL~\cite{de2010grammatical} merges states based on the minimum description length principle, but only works with positive traces. The k-tail algorithm~\cite{biermann1972synthesis} and its extensions~\cite{lorenzoli2008automatic, krka2014automatic} merge states based on trace suffixes and incorporate input predicates to capture data relations. Krka et al.~\cite{krka2014automatic} developed TEMI to mine more complex EFSMs, while Synoptic~\cite{beschastnikh2011leveraging} uses temporal invariants before applying k-tail.

Other methods such as \contractor~\cite{de2010automated, krka2014automatic} infer FSMs from program invariants, and Le and Lo~\cite{le2018deep} use deep learning for automata specification generation. Despite the success of these techniques, smart contract specification mining presents unique challenges, particularly the dynamic, stateful environment of real-world smart contracts. The interaction between users and the system, reflected in transaction histories, requires specialized techniques like trace slicing to extract meaningful specifications. This complexity differentiates smart contract mining from traditional mining approaches.

%% file: sections/Conclusion.tex
\section{Conclusion}\label{sec:conclusion}

In this paper, we have formally defined the specification mining problem for smart contracts and
proposed a CEGAR-like approach to mine automata specifications based on past transaction histories.
The mined specifications capture not only the allowed function invocation sequences, but also the
inferred program invariants describing contract semantics precisely.
Such contract specifications are useful in contract understanding, testing, verification, and
validation.
Our evaluation results show that our tool, \tool, mines specifications accurately and efficiently;
it may also be used to enhance symbolic analysis for smart contracts and facilitate developer in maintaining high-quality document and test suites.

%% file: sections/Illustration.tex
\appendix
\subsection{Illustrations}\label{subsec:illustrationexample}

\begin{figure*}
	\begin{subfigure}[b]{.2\textwidth}
		\centering
		 \includegraphics[width=.9\columnwidth]{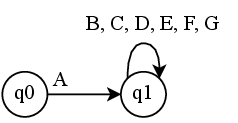}
		\caption{The initial FSM.}\label{fig:initialautomaton}
	\end{subfigure}
	\hspace{-1mm}
	\begin{subfigure}[b]{.3\textwidth}
		\centering
		\includegraphics[width=0.9\columnwidth]{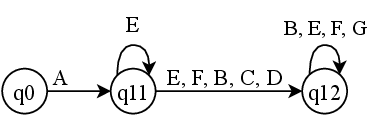}
		\caption{The second FSM.}\label{fig:secondautomaton}
	\end{subfigure}
	\hspace{-1mm}
	\begin{subfigure}[b]{.4\textwidth}
		\centering
		\includegraphics[width=.9\columnwidth]{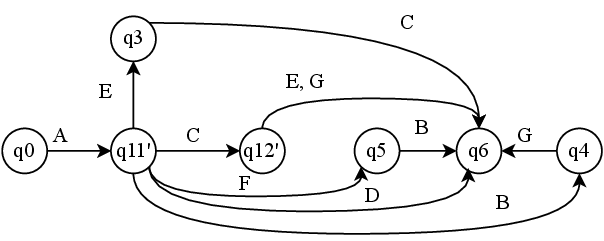}
		\caption{The final FSM.}\label{fig:finalautomaton}
	\end{subfigure}
	\caption{The illustration of automata mining process.}
\end{figure*}


We use \texttt{GameChannel} as an example.
Initially, rule \textsc{Init} is applied.
As shown in~\cref{fig:initialautomaton}, the resulting automaton has two states: $q_0$ and $q_1$.
For the initial state $q_0$, all state variables are valued zero, namely $P_1 \land P_6 \land P_8
\land P_{10}$, while $q_1$ covers all the remaining values of the state variables, namely $\lnot
(P_1 \land P_6 \land P_8 \land P_{10}) $.
For each function $m$, its precondition $g_m$ is added as the
guard function set $G$, and the post-condition $u_m$ is added
to the update function set $U$, which we elide in these diagrams for simplicity.

Rule \textsc{Construct} adds all possible transitions to the automaton.
It is easy to see in~\cref{fig:initialautomaton} that the resulting automaton contains all the
observed concrete function invocation sequences.
However, spurious behaviors may have been allowed in the automaton.
For example, the symbolic path $q_0-A-q_1-D-q_1-C-q_1$ is spurious, as it has no supported concrete
invocation sequence within the observations.
Therefore, we need to apply rule \textsc{RmPath} to eliminate such spurious paths.
In rule \textsc{RmPath}, we use the function \texttt{concretize} on a symbolic path to represent
one of its possible concrete invocation sequences.
For the spurious path $q_0-A-q_1-D-q_1-C-q_1$,
the automaton needs to refine itself by either splitting the state $q_1$ ($q_n$ in
Line~\ref{line:sat-q}) or removing the transition $q_1-C-q_1$ ($t_{n+1}$ in Line~\ref{line:minust})
using~\cref{algo:split-remove}.


\Cref{algo:split-remove} shows our approach to remove a spurious path via state splitting or transition removal.
For the state $q_1$ and the transition $q_1-C-q_1$ in~\cref{fig:initialautomaton},
we first attempt to split $q_1$ with the precondition of function $C$, namely $P_2 \lor (P_3 \land P_6)$ of \texttt{serverCancelActiveGame}.
We use a decision procedure $\mathrm{SAT}$ (Line~\ref{line:sat}) to decide if a symbolic state is
feasible.
Due to the fact that the two newly created states (Lines~\ref{line:split-1}--\ref{line:split-2}) for $q_1$ are feasible,
we replace $q_1$ with the new states (Line~\ref{line:sat-q}) and remove the outdated transitions starting or ending with $q_1$ (Line~\ref{line:sat-t}).
The modified automata structure will be returned (Line~\ref{line:return}) and when applying rule \textsc{Construct} again, we can obtain such new automaton in~\cref{fig:secondautomaton} that has three states.
In~\cref{algo:split-remove}, it is possible that $q_n$ is not splittable with the function precondition (Lines~\ref{line:not-sat-begin}--\ref{line:not-sat-end}).
When $t_{n+1}$ has never been visited by replaying all the concrete invocation sequences on the automaton (Line~\ref{line:notreachable}), we safely remove $t_{n+1}$ to eliminate the spurious path.
Otherwise, we will split $q_n$ with an equivalent predicate representing a partial set of the visited concrete states (such as $s_0$ to $s_{17}$ of~\cref{tab:Dicether-LOG}) along a symbolic path $\pi'$ (Lines~\ref{line:other-start}--\ref{line:other-split})
and also remove the outdated transitions relevant to $q_n$ (Line~\ref{line:other-remove}).


\Cref{fig:secondautomaton} also contains spurious paths, e.\,g., $q_0-A-q_{11}-E-q_{11}-E-q_{12}$,
which have to be eliminated.
By applying the aforementioned rules many times, finally we are able to mine an automaton as shown
in~\cref{fig:finalautomaton}, which has the exact automata structure as \cref{fig:dicether-spec}
except the guard condition that we elide for simplicity.
Apart from the common usage scenarios, we can also understand the root semantics of each state in
the automaton.
Note that each symbolic state is represented by a set of predicates.
As described before, $q_0$ is represented by $P_1 \land P_6 \land P_8 \land P_{10}$.
In the automaton,
$q_{11}'$ is represented by $P_2$ that says \texttt{status} is \texttt{ACTIVE}, $q_{12}'$ is represented by $P_4 \land P_6$ that says \texttt{status} is \texttt{SERVER\_INITIATED\_END} and \texttt{roundId} is zero, $q_3$ is represented by $P_3\land P_6$ that says \texttt{status} is \texttt{USER\_INITIATED\_END} and \texttt{roundId} is zero,
$q_4$ is represented by $P_4 \land P_7$ that says \texttt{status} is \texttt{SERVER\_INITIATED\_END} but \texttt{roundId} is larger than zero, $q_5$ is represented by $P_3 \land P_7$ that says \texttt{status} is \texttt{USER\_INITIATED\_END} but \texttt{roundId} is larger than zero,
and $q_6$ having no outcoming transitions suggests it is the final state whose predicates cover all the remaining cases.
We believe such automata precisely model smart contract semantics, which is helpful
for understanding contract behaviors.

\subsection{Termination and Soundness}\label{subsec:soundness}

We now discuss the termination condition of our specification mining rules shown in
\cref{fig:rules}.
To ensure the termination of our algorithm, the applications of the rules have to follow a
\emph{fair scheduling}.
A fair scheduling is an infinite sequence of actions: $a_1, a_2, \ldots, a_n$, where $a_i \in
\{\textsc{Construct}, \textsc{RmPath}\}$ and the following conditions apply:
\begin{enumerate}[leftmargin=*,label=(\arabic*)]
\item \textsc{Construct} is the initial action,
\item if new states are added, \textsc{Construct} keeps applying until it is inapplicable,
\item \textsc{RmPath} appears infinitely often.
\end{enumerate}

Condition (1) enables the initial construction of the automaton, and Condition (2) ensures the
integrity of any subsequent automata construction.
Condition (3) ensures that every spurious path of the automaton will eventually be removed.
We are now ready to prove the termination of our specification mining algorithm with the
aforementioned fair conditions.
We assume that the algorithm terminates when \textsc{RmPath} and \textsc{Construct} are both no
longer applicable.

\begin{theorem}\label{thm:termination}
Our specification mining algorithm terminates after a finite number of actions in any fair
scheduling.
\end{theorem}

\begin{proof}
	For \history which is equivalent to the observed concrete invocation sequences, let $N_s$ be the number of concrete states and $N_p$ be the number of all concrete invocation sequences.
	For each application of rule \textsc{RmPath},
	it will take no more than $N_p + 1 $ times to traverse the resulting automaton to find such a spurious path $\pi$ if it exists.
	Therefore, \textsc{RmPath} itself is terminated.

	As discussed in~\cref{algo:split-remove},
	a transition may be removed from the resulting automaton.
	According to the fairness conditions,
	\textsc{RmPath} will keep applying until it splits a state to generate two new fine-grained states.
	Suppose the current automaton has $n$ symbolic states and $k$ symbolic state transitions.
	We easily have $k \le n \times |M|\times n$, where $|M|$ is the number of contract functions.
	Therefore, the times for \textsc{RmPath} to find one more reachable states is no more than $n \times |M|\times  n$.
	So it takes no more than $N_s \times \max(n) \times |M|\times \max(n)$ times to find all reachable states, where $\max(n) \leq N_s$.

	Next, we prove that the number of \textsc{RmPath} actions in any scheduling cannot exceed
  $(N_s + 1) \times \max(n) \times |M| \times \max(n)$.
	We know that finding all reachable states does not exceed $N_s \times \max(n) \times |M| \times  \max(n)$.
	After that, \textsc{RmPath} will be applied to remove any unreachable transition if applicable.
	We know this will apply no more than $\max(n) \times |M| \times \max(n)$ times.
	Finally, \textsc{RmPath} will be inapplicable and the algorithm will terminate with no more than $(N_s + 1 )\times \max(n) \times |M| \times \max(n)$ actions of \textsc{RmPath}.
	So, \cref{thm:termination} is proved.
\end{proof}

Under the fair scheduling, we prove the soundness of the resulting automaton.

\begin{theorem}\label{thm:soundness}
  The resulting automaton (ignoring loop transitions) is a \minimalabstraction of \history.
\end{theorem}
\begin{proof}
The resulting automaton is an \efsmtext.
To prove \cref{thm:soundness}, we need to prove that the automaton satisfies Condition (1) and
Condition (2) in \cref{def:minimal}.

Condition (1) holds because by default, all uninitialized state variables of smart contracts are
valued zero, and $s_0$ is the initial concrete state where all state variables are valued zero.
The automaton takes the zero-valued predicates of all state variables as the initial symbolic state
$q_0$.
In this setting, $\{s_0\} = \alpha^{-1}(q_0)$.
Thus, Condition (1) is satisfied.

Condition (2) can be verified by proving: (i) every observed concrete invocation sequence (i.e., a
concrete path in \history) is preserved in the resulting automaton;
(ii) every symbolic path of the resulting automaton has a corresponding concrete path in \history,
i.\,e., no symbolic path is spurious.
According to rule \textsc{Construct}, (i) holds in the resulting automaton.
We prove (ii) by contradiction.
Suppose the resulting automaton contains a spurious path $\pi: q_0 t_1 t_2 \cdots q_n \in$ \efsm
where $\nexists \;\mathrm{concretize}(\pi) \in $ \history.
The \textsc{RmPath} rule will be applied to remove $\pi$ either via state splitting or transition
removal at finite actions.
Thus, there is no such spurious path $\pi$, which is in contradiction with the assumption.
Therefore, (b) is proved.
Hence, Condition (2) is satisfied by the resulting automaton.

With Condition (1) and (2) satisfied, \cref{thm:soundness} is proved.
\end{proof}

%% file: main.bbl
\begin{thebibliography}{10}
\providecommand{\url}[1]{#1}
\csname url@samestyle\endcsname
\providecommand{\newblock}{\relax}
\providecommand{\bibinfo}[2]{#2}
\providecommand{\BIBentrySTDinterwordspacing}{\spaceskip=0pt\relax}
\providecommand{\BIBentryALTinterwordstretchfactor}{4}
\providecommand{\BIBentryALTinterwordspacing}{\spaceskip=\fontdimen2\font plus
\BIBentryALTinterwordstretchfactor\fontdimen3\font minus
  \fontdimen4\font\relax}
\providecommand{\BIBforeignlanguage}[2]{{%
\expandafter\ifx\csname l@#1\endcsname\relax
\typeout{** WARNING: IEEEtran.bst: No hyphenation pattern has been}%
\typeout{** loaded for the language `#1'. Using the pattern for}%
\typeout{** the default language instead.}%
\else
\language=\csname l@#1\endcsname
\fi
#2}}
\providecommand{\BIBdecl}{\relax}
\BIBdecl

\bibitem{nakamoto2008bitcoin}
S.~Nakamoto, ``Bitcoin: A peer-to-peer electronic cash system,''
  \emph{Decentralized Business Review}, p. 21260, 2008.

\bibitem{Ethereum}
G.~Wood, ``{Ethereum}: A secure decentralised generalised transaction ledger,''
  \emph{Ethereum project yellow paper}, vol. 151, pp. 1--32, 2014.

\bibitem{dappradar}
``Dappradar,'' \url{https://dappradar.com/}, 2023.

\bibitem{Etherscan}
``Etherscan,'' \url{https://etherscan.io}, 2023.

\bibitem{durieux2020empirical}
T.~Durieux, J.~F. Ferreira, R.~Abreu, and P.~Cruz, ``Empirical review of
  automated analysis tools on 47,587 {Ethereum} smart contracts,'' in
  \emph{Proceedings of the ACM/IEEE 42nd International conference on software
  engineering}, 2020, pp. 530--541.

\bibitem{eip20}
``{EIP-20}: A standard interface for tokens,''
  \url{https://eips.ethereum.org/EIPS/eip-20}, 2015.

\bibitem{chen2019tokenscope}
T.~Chen, Y.~Zhang, Z.~Li, X.~Luo, T.~Wang, R.~Cao, X.~Xiao, and X.~Zhang,
  ``{TokenScope}: Automatically detecting inconsistent behaviors of
  cryptocurrency tokens in ethereum,'' in \emph{Proceedings of the 2019 ACM
  SIGSAC conference on computer and communications security}, 2019, pp.
  1503--1520.

\bibitem{qin2021attacking}
K.~Qin, L.~Zhou, B.~Livshits, and A.~Gervais, ``Attacking the {DeFi} ecosystem
  with flash loans for fun and profit,'' in \emph{International Conference on
  Financial Cryptography and Data Security}.\hskip 1em plus 0.5em minus
  0.4em\relax Springer, 2021, pp. 3--32.

\bibitem{jiao2020semantic}
J.~Jiao, S.~Kan, S.-W. Lin, D.~Sanan, Y.~Liu, and J.~Sun, ``Semantic
  understanding of smart contracts: Executable operational semantics of
  {Solidity},'' in \emph{2020 IEEE Symposium on Security and Privacy
  (SP)}.\hskip 1em plus 0.5em minus 0.4em\relax IEEE, 2020, pp. 1695--1712.

\bibitem{Liu2022IAD}
Y.~Liu and Y.~Li, ``{InvCon}: A dynamic invariant detector for {Ethereum} smart
  contracts,'' in \emph{Proceedings of the 37th IEEE/ACM International
  Conference on Automated Software Engineering (ASE)}, Oct. 2022.

\bibitem{Wang2019VUL}
H.~Wang, Y.~Li, S.-W. Lin, L.~Ma, and Y.~Liu, ``{VULTRON}: Catching vulnerable
  smart contracts once and for all,'' in \emph{Proceedings of the 41st
  International Conference on Software Engineering: New Ideas and Emerging
  Results (ICSE-NIER)}.\hskip 1em plus 0.5em minus 0.4em\relax IEEE Press, 5
  2019, pp. 1--4.

\bibitem{permenev2020verx}
A.~Permenev, D.~Dimitrov, P.~Tsankov, D.~Drachsler-Cohen, and M.~Vechev,
  ``Verx: Safety verification of smart contracts,'' in \emph{2020 IEEE
  symposium on security and privacy (SP)}.\hskip 1em plus 0.5em minus
  0.4em\relax IEEE, 2020, pp. 1661--1677.

\bibitem{liu2020towards}
Y.~Liu, Y.~Li, S.-W. Lin, and R.~Zhao, ``Towards automated verification of
  smart contract fairness,'' in \emph{Proceedings of the 28th ACM Joint Meeting
  on European Software Engineering Conference and Symposium on the Foundations
  of Software Engineering}, 2020, pp. 666--677.

\bibitem{li2020securing}
A.~Li, J.~A. Choi, and F.~Long, ``Securing smart contract with runtime
  validation,'' in \emph{Proceedings of the 41st ACM SIGPLAN Conference on
  Programming Language Design and Implementation}, 2020, pp. 438--453.

\bibitem{mavridou2018designing}
A.~Mavridou and A.~Laszka, ``Designing secure ethereum smart contracts: A
  finite state machine based approach,'' in \emph{International Conference on
  Financial Cryptography and Data Security}.\hskip 1em plus 0.5em minus
  0.4em\relax Springer, 2018, pp. 523--540.

\bibitem{mavridou2018tool}
------, ``Tool demonstration: {FSolidM} for designing secure {Ethereum} smart
  contracts,'' in \emph{International conference on principles of security and
  trust}.\hskip 1em plus 0.5em minus 0.4em\relax Springer, 2018, pp. 270--277.

\bibitem{Liu2020MAM}
Y.~Liu, Y.~Li, S.-W. Lin, and Q.~Yan, ``{ModCon}: A model-based testing
  platform for smart contracts,'' in \emph{Proceedings of the 28th ACM Joint
  European Software Engineering Conference and Symposium on the Foundations of
  Software Engineering (FSE)}, Nov. 2020.

\bibitem{mavridou2019verisolid}
A.~Mavridou, A.~Laszka, E.~Stachtiari, and A.~Dubey, ``{VeriSolid}:
  Correct-by-design smart contracts for {Ethereum},'' in \emph{International
  Conference on Financial Cryptography and Data Security}.\hskip 1em plus 0.5em
  minus 0.4em\relax Springer, 2019, pp. 446--465.

\bibitem{wang2019formal}
Y.~Wang, S.~K. Lahiri, S.~Chen, R.~Pan, I.~Dillig, C.~Born, I.~Naseer, and
  K.~Ferles, ``Formal verification of workflow policies for smart contracts in
  azure blockchain,'' in \emph{Working Conference on Verified Software:
  Theories, Tools, and Experiments}.\hskip 1em plus 0.5em minus 0.4em\relax
  Springer, 2019, pp. 87--106.

\bibitem{aarts2012automata}
F.~Aarts, F.~Heidarian, H.~Kuppens, P.~Olsen, and F.~Vaandrager, ``Automata
  learning through counterexample guided abstraction refinement,'' in
  \emph{International Symposium on Formal Methods}.\hskip 1em plus 0.5em minus
  0.4em\relax Springer, 2012, pp. 10--27.

\bibitem{clarke2000counterexample}
E.~Clarke, O.~Grumberg, S.~Jha, Y.~Lu, and H.~Veith, ``Counterexample-guided
  abstraction refinement,'' in \emph{International Conference on Computer Aided
  Verification}.\hskip 1em plus 0.5em minus 0.4em\relax Springer, 2000, pp.
  154--169.

\bibitem{de2010grammatical}
C.~De~la Higuera, \emph{Grammatical inference: learning automata and
  grammars}.\hskip 1em plus 0.5em minus 0.4em\relax Cambridge University Press,
  2010, vol.~24, no. 3-4.

\bibitem{le2018deep}
T.-D.~B. Le and D.~Lo, ``Deep specification mining,'' in \emph{Proceedings of
  the 27th ACM SIGSOFT International Symposium on Software Testing and
  Analysis}, 2018, pp. 106--117.

\bibitem{biermann1972synthesis}
A.~W. Biermann and J.~A. Feldman, ``On the synthesis of finite-state machines
  from samples of their behavior,'' \emph{IEEE transactions on Computers}, vol.
  100, no.~6, pp. 592--597, 1972.

\bibitem{krka2014automatic}
I.~Krka, Y.~Brun, and N.~Medvidovic, ``Automatic mining of specifications from
  invocation traces and method invariants,'' in \emph{Proceedings of the 22nd
  ACM SIGSOFT International Symposium on Foundations of Software Engineering},
  2014, pp. 178--189.

\bibitem{lorenzoli2008automatic}
D.~Lorenzoli, L.~Mariani, and M.~Pezz{\`e}, ``Automatic generation of software
  behavioral models,'' in \emph{Proceedings of the 30th international
  conference on Software engineering}, 2008, pp. 501--510.

\bibitem{graf1997construction}
S.~Graf and H.~Saidi, ``Construction of abstract state graphs with {PVS},'' in
  \emph{Computer Aided Verification}, vol.~97, 1997, pp. 72--83.

\bibitem{daikon}
``Daikon,'' \url{http://plse.cs.washington.edu/daikon/}, 2021, the Daikon
  invariant detector.

\bibitem{lee2011mining}
C.~Lee, F.~Chen, and G.~Ro{\c{s}}u, ``Mining parametric specifications,'' in
  \emph{Proceedings of the 33rd International Conference on Software
  Engineering}, 2011, pp. 591--600.

\bibitem{beillahi2020behavioral}
S.~M. Beillahi, G.~Ciocarlie, M.~Emmi, and C.~Enea, ``Behavioral simulation for
  smart contracts,'' in \emph{Proceedings of the 41st ACM SIGPLAN Conference on
  Programming Language Design and Implementation}, 2020, pp. 470--486.

\bibitem{cheng1993automatic}
K.-T. Cheng and A.~S. Krishnakumar, ``Automatic functional test generation
  using the extended finite state machine model,'' in \emph{30th ACM/IEEE
  Design Automation Conference}.\hskip 1em plus 0.5em minus 0.4em\relax IEEE,
  1993, pp. 86--91.

\bibitem{chauhan2002automated}
P.~Chauhan, E.~Clarke, J.~Kukula, S.~Sapra, H.~Veith, and D.~Wang, ``Automated
  abstraction refinement for model checking large state spaces using sat based
  conflict analysis,'' in \emph{International Conference on Formal Methods in
  Computer-Aided Design}.\hskip 1em plus 0.5em minus 0.4em\relax Springer,
  2002, pp. 33--51.

\bibitem{dicether}
``Dicether: A secure dice game,'' https://dicether.github.io/paper/paper.pdf,
  2018.

\bibitem{liu2024automated}
Y.~Liu, C.~Zhang \emph{et~al.}, ``Automated invariant generation for solidity
  smart contracts,'' \emph{arXiv preprint arXiv:2401.00650}, 2024.

\bibitem{moura2008z3}
L.~d. Moura and N.~Bj{\o}rner, ``{Z3}: An efficient {SMT} solver,'' in
  \emph{International conference on Tools and Algorithms for the Construction
  and Analysis of Systems}.\hskip 1em plus 0.5em minus 0.4em\relax Springer,
  2008, pp. 337--340.

\bibitem{SuperRare}
``{SuperRare},'' \url{https://www.dapp.com/app/SuperRare}, 2022.

\bibitem{MoonCatRescue}
``{MoonCatRescue},'' \url{https://dappradar.com/ethereum/games/mooncatrescue},
  2022.

\bibitem{QuickNode}
``Quicknode,'' \url{https://www.quicknode.com/}, 2023.

\bibitem{lo2006quark}
D.~Lo and S.-C. Khoo, ``Quark: Empirical assessment of automaton-based
  specification miners,'' in \emph{2006 13th Working Conference on Reverse
  Engineering}.\hskip 1em plus 0.5em minus 0.4em\relax IEEE, 2006, pp. 51--60.

\bibitem{le2015synergizing}
T.-D.~B. Le, X.-B.~D. Le, D.~Lo, and I.~Beschastnikh, ``Synergizing
  specification miners through model fissions and fusions (t),'' in \emph{2015
  30th IEEE/ACM International Conference on Automated Software Engineering
  (ASE)}.\hskip 1em plus 0.5em minus 0.4em\relax IEEE, 2015, pp. 115--125.

\bibitem{lo2012learning}
D.~Lo, L.~Mariani, and M.~Santoro, ``Learning extended fsa from software: An
  empirical assessment,'' \emph{Journal of Systems and Software}, vol.~85,
  no.~9, pp. 2063--2076, 2012.

\bibitem{de2010automated}
G.~De~Caso, V.~Braberman, D.~Garbervetsky, and S.~Uchitel, ``Automated
  abstractions for contract validation,'' \emph{IEEE Transactions on Software
  Engineering}, vol.~38, no.~1, pp. 141--162, 2010.

\bibitem{guth2018specification}
F.~Guth, V.~W{\"u}stholz, M.~Christakis, and P.~M{\"u}ller, ``Specification
  mining for smart contracts with automatic abstraction tuning,'' \emph{arXiv
  preprint arXiv:1807.07822}, 2018.

\bibitem{issue278}
``Bug report in defective-component-counter smart contract,''
  \url{https://github.com/Azure-Samples/blockchain/issues/278}, 2024.

\bibitem{issue279}
``Bug report in digital-locker smart contract,''
  \url{https://github.com/Azure-Samples/blockchain/issues/279}, 2024.

\bibitem{issue280}
``Bug report in hello-blockchain smart contract,''
  \url{https://github.com/Azure-Samples/blockchain/issues/280}, 2024.

\bibitem{issue281}
``Bug report in simple-marketplace smart contract,''
  \url{https://github.com/Azure-Samples/blockchain/issues/281}, 2024.

\bibitem{learnlib}
``{LearnLib}--an open framework for automata learning,''
  \url{https://learnlib.de/}, 2022.

\bibitem{mythril}
``{Mythril},'' \url{https://github.com/ConsenSys/mythril}, 2019, a Security
  Analysis Tool for EVM Bytecode.

\bibitem{wang2020oracle}
H.~Wang, Y.~Liu, Y.~Li, S.-W. Lin, C.~Artho, L.~Ma, and Y.~Liu,
  ``Oracle-supported dynamic exploit generation for smart contracts,''
  \emph{IEEE Transactions on Dependable and Secure Computing}, vol.~19, no.~3,
  pp. 1795--1809, 2020.

\bibitem{grieco2020echidna}
G.~Grieco, W.~Song, A.~Cygan, J.~Feist, and A.~Groce, ``Echidna: effective,
  usable, and fast fuzzing for smart contracts,'' in \emph{Proceedings of the
  29th ACM SIGSOFT international symposium on software testing and analysis},
  2020, pp. 557--560.

\bibitem{shou2023ityfuzz}
C.~Shou, S.~Tan, and K.~Sen, ``Ityfuzz: Snapshot-based fuzzer for smart
  contract,'' in \emph{Proceedings of the 32nd ACM SIGSOFT International
  Symposium on Software Testing and Analysis}, 2023, pp. 322--333.

\bibitem{yang2024exploring}
M.~Yang, D.~Lie, and N.~Papernot, ``Exploring strategies for guiding symbolic
  analysis with machine learning prediction,'' in \emph{2024 IEEE International
  Conference on Software Analysis, Evolution and Reengineering (SANER)}.\hskip
  1em plus 0.5em minus 0.4em\relax IEEE, 2024, pp. 659--669.

\bibitem{cryptopunks}
``Cryptopunks,'' \url{https://www.larvalabs.com/cryptopunks/}, 2023.

\bibitem{cryptopunkgittest}
``{CryptoPunks}: Collectible characters on the {Ethereum} blockchain,''
  \url{https://github.com/larvalabs/cryptopunks/tree/master/test}, 2017.

\bibitem{tan2022soltype}
B.~Tan, B.~Mariano, S.~K. Lahiri, I.~Dillig, and Y.~Feng, ``Soltype: refinement
  types for arithmetic overflow in solidity,'' \emph{Proceedings of the ACM on
  Programming Languages}, vol.~6, no. POPL, pp. 1--29, 2022.

\bibitem{liu2022learning}
J.~Liu, Y.~Chen, B.~Tan, I.~Dillig, and Y.~Feng, ``Learning contract invariants
  using reinforcement learning,'' in \emph{Proceedings of the 37th IEEE/ACM
  International Conference on Automated Software Engineering}, 2022, pp. 1--11.

\bibitem{wang2024smartinv}
S.~J. Wang, K.~Pei, and J.~Yang, ``Smartinv: Multimodal learning for smart
  contract invariant inference,'' in \emph{2024 IEEE Symposium on Security and
  Privacy (SP)}.\hskip 1em plus 0.5em minus 0.4em\relax IEEE Computer Society,
  2024, pp. 126--126.

\bibitem{walkinshaw2016inferring}
N.~Walkinshaw, R.~Taylor, and J.~Derrick, ``Inferring extended finite state
  machine models from software executions,'' \emph{Empirical Software
  Engineering}, vol.~21, no.~3, pp. 811--853, 2016.

\bibitem{beschastnikh2011leveraging}
I.~Beschastnikh, Y.~Brun, S.~Schneider, M.~Sloan, and M.~D. Ernst, ``Leveraging
  existing instrumentation to automatically infer invariant-constrained
  models,'' in \emph{Proceedings of the 19th ACM SIGSOFT symposium and the 13th
  European conference on Foundations of software engineering}, 2011, pp.
  267--277.

\bibitem{gold1967language}
E.~M. Gold, ``Language identification in the limit,'' \emph{Information and
  control}, vol.~10, no.~5, pp. 447--474, 1967.

\bibitem{angluin1987learning}
D.~Angluin, ``Learning regular sets from queries and counterexamples,''
  \emph{Information and computation}, vol.~75, no.~2, pp. 87--106, 1987.

\end{thebibliography}
